\documentclass[aps,prl,showpacs,twocolumn,superscriptaddress]{revtex4-2}
\usepackage{amsmath,amssymb}
\usepackage{graphicx}
\usepackage{dcolumn}
\usepackage{bm}
\usepackage{dsfont}
\usepackage{amsthm}
\newtheorem{lemma}{Lemma}[section]

\begin{document}
\title{Chiral state conversion near an exceptional point: speed-noise competition}

\author{Qing-Wei Wang}
\email[]{qingweiwang2012@163.com}
\affiliation{School of Information Engineering, Zhejiang Ocean University, Zhoushan, Zhejiang 316022, China}

\date{\today}

\begin{abstract}
  One intriguing property of non-Hermitian systems is the  breakdown of adiabatic theorem and chiral state conversion as the system dynamically encircles exceptional points. However, the subtle dependence of the chiral dynamics on the loop geometry, the starting point, the encircling speed and especially the noise has not been studied systematically. Here we propose a non-chirality degree $\chi_c$ to measure the chirality quantitatively and analyze it in dynamics without noise by exact solution and dynamics with noise by numerical integration. The exact dynamics starting from the broken phase show chirality oscillations, which are extremely sensitive to noise when the speed is small. The encircling speed and the noise strength are found to compete with each other in determining $\chi_c$, resulting in two distinguished limits, namely the noisy limit and the clean limit. The critical boundary between the two limits satisfies a simple scaling law, which could be explained in terms of first-order perturbation theory and the condition number of the transfer matrix. Our findings reveal the essential role played by noise in non-Hermitian dynamics and are relevant for both theoretical and experimental investigations.  
\end{abstract}
\maketitle
{\it Introduction.---} Non-Hermitian systems supports a new kind of degeneracies, known as exceptional points (EPs) \cite{PhysRevE.61.929,Graefe_2008,GARRISON1988177, Keck_2003,Berry2004}, where both eigenvalues and eigenvectors coalesce. This results in a number of intriguing properties such as loss-induced transparency \cite{PhysRevLett.103.093902}, single-mode lasing \cite{Feng972,Hodaei975}, unidirectional invisibility \cite{PhysRevLett.106.213901,Feng2013}, to mention a few. 
Of particular interest is the topological properties associated with the quasistatic encirclement of an EP.
It was found that the instantaneous eigenstates swap with each other at the end of the parameter cycle with one acquiring a geometric phase \cite{PhysRevA.72.014104,Heiss_2012}. This can be attributed to the branch point character of the EP and has been observed in microwave cavities \cite{PhysRevLett.86.787} and exciton-polariton systems \cite{nature526-554}.
However, dynamical evolution around an EP is drastically different from the quasistatic encirclement \cite{Uzdin_2011, PhysRevA.88.010102,PhysRevA.88.033842, PhysRevLett.99.173003, PhysRevLett.103.123003, PhysRevLett.116.133903, PhysRevLett.118.040401,  PhysRevLett.118.093002,PhysRevX.8.021066,PhysRevLett.125.187403}, leading to chiral state conversion in the sense that different encircling directions result in different output states. These surprising effects were recently observed in microwave \cite{nature537-76}, optomechanical  systems \cite{nature537-80}, coupled pendulums \cite{EVEN2024118239, rspa.2024.0335}, coupled waveguides \cite{nature562-86, LSA8-88, CP2-63, PhysRevLett.124.153903, PhysRevLett.129.127401, nc13-2123, LSA14-77, APR12-041409_2025, APL127-122201_2025}, and electric circuits \cite{CP3-140_2020,PhysRevLett.124.133905, lpor.202100675, nwac259_2023}. 

Some theoretical and experimental studies, however, have shown that the chiral state transfer is not necessarily associated with EP-encirclement \cite{PhysRevA.96.052129, PhysRevA.102.040201, PhysRevA.103.023531, nature605-256_2022, PhysRevResearch.5.033053, LSA13-65}. For example, chiral behavior can also be observed without encircling any EP while nonchiral behavior could be observed in EP-encircling dynamics, and dynamically encircling an EP along two homotopic loops in the parameter space may result in distinct outcomes when multiple exceptional points are involved \cite{PhysRevA.99.063831, nc9-4808_2018}. The chirality is found to depend not only on the loop geometry, but also on the starting point \cite{PhysRevX.8.021066,CP2-63} and the encircling speed (or the degree of adiabaticity) \cite{CP3-140_2020, nature605-256_2022,PhysRevResearch.5.033053, PhysRevLett.128.160401, APLQ1-046107_2024,  PRXQuantum.6.020328}. In addition, someone demonstrated robustness of the chirality with respect to noise \cite{CP8_91_2025, PhysRevResearch.7.013159}, while some others \cite{arxiv250204214} showed that the presence of noise is essential and would drastically alters the dynamics of non-Hermitian systems. However, systematic studies of the chirality in the noise-speed parameter space is till lacking. 

One important issue here is to search for general rules or universal relations in the chiral/nonchiral state conversion process, which has a complicated dependence on the loop topology, starting point and dynamical parameters such as the encircling speed, noise or dissipation. Previous efforts have mainly focused on the loop topology with multiple EPs \cite{PhysRevLett.123.066405, PhysRevLett.127.253901,PhysRevLett.130.157201, CP7_109_2024, NC15_1369_2024} in the adiabatic limit, and no general features have been found concerning the dynamical parameters's effect on the chirality.   

In this Letter, we theoretically analyze the state conversion processes in cyclic dynamics of a non-Hermitian two-level system near its EPs, focusing on the dependence of the chirality measured by a \emph{non-chirality degree} on the dynamical parameters. The state conversion processes show chirality oscillations in the absence of noise for loops starting from points where the Hamiltonian has imaginary eigenvalues, and such oscillations are extremely sensitive to noise. So we systematically investigate the noise effect by adding a white noise term in the Hamiltonian, and then examine the dependence of chirality on the loop radius, the starting point, the speed and the noise strength. We discover a general speed-noise competition behavior and its critical boundary scaling relation, which is in turn explained through a first-order perturbation theory. Our findings are relevant to both experimental and theoretical studies on non-Hermitian dynamics.

{\it Theoretical Model and Symmetries.---}We consider a two-state system governed by the evolution equation $i\hbar\partial_t|\psi(t)\rangle= H(t) |\psi(t)\rangle$, where the time-dependent Hamiltonian is given by
\begin{equation}\label{eq:Hamiltonian}
  H(t)=\kappa \sigma_x + h_z(t)\sigma_z,
\end{equation}
with the state vector $|\psi(t)\rangle=(a(t),b(t))^T$. The coupling strength $\kappa$ is kept constant while the time-varying $h_z(t)=\delta(t) + ig(t)$ is constructed from the gain/loss $g(t)$ and level detuning $\delta(t)$.
For simplicity, we would take the natural units $\hbar=\kappa=1$ in the following. Then two EPs are established in the parameter space at $g=\pm1$ and $\delta=0$, where the eigenvalues coalesce with the corresponding eigenvectors collapsing to $(\pm i,1)^T$. 
This Hamiltonian has real (imaginary) eigenvalues for $\delta=0$ and $|g|<1 (>1)$, which would be referred to as symmetric (broken) phase\cite{PhysRevX.8.021066}. 


Now we consider a class of exactly solvable circular loops\cite{PhysRevLett.118.093002,PhysRevX.8.021066} described by 
\begin{equation}\label{eq:loops}
  h_z(t)=i \left(g_0-\rho\, e^{i(\omega t+\theta_i)} \right),
\end{equation}
where $\omega$ denotes the angular velocity of the encircling and $\rho$ represents the radius of the circle centered at $(\delta=0, g=g_0)$. The equation (\ref{eq:loops}) represents a counterclockwise (CCW) loop if $\omega>0$, with the starting point described by an angle $\theta_i$. A clockwise (CW) loop could be obtained by making the replacement $\omega\rightarrow -\omega$. The dynamics is conveniently described by the transfer matrix $S(\theta_f,\theta_i)$ in terms of 
\begin{equation}\label{eq:evolution:S}
  [a(t),b(t)]^T =S(\theta_f,\theta_i)\, [a(0),b(0)]^T,
\end{equation}
where $\theta_f=\theta_i+\omega t$ locates the final point of the loop. 

It's instructive to first analyse the symmetry of the transfer matrix before solving it exactly \cite{*[{See Supplemental Material for symmetry analysis of the transfer matrix, exact solution, asymmetric analysis, Floquet analysis, and more numerical results.}][{}] sm}. 
We can show that 
(i) $S_{\text{CCW}}(\theta_f,\theta_i)=\sigma_z  S_{\text{CW}}(-\theta_f,-\theta_i)^\ast \sigma_z$, and 
(ii) $S_{\text{CCW}}(\theta_f,\theta_i)=\left[ S_{\text{CW}}(\theta_i,\theta_f) \right]^T$,
where the subscript CCW/CW indicates the loop direction, and 
(iii) $S(\theta_f,\theta_i)=\sigma_z  \left[S(-\theta_i,-\theta_f)\right]^\dag \sigma_z$, where $S$ without subscript indicates validity for both CCW and CW loops. Specially, for $\theta_i=0$ or $\pi$ and $\theta_f=\theta_i\pm2\pi$, property (iii) gives $S=\sigma_zS^\dag \sigma_z$, where $S$ denotes either $S_{\text{CCW}}(2\pi,0)$, $S_{\text{CCW}}(\pi,-\pi)$, $S_{\text{CW}}(0,2\pi)$ or $S_{\text{CW}}(-\pi,\pi)$.  This symmetry tells us that $S_{11}$ and $S_{22}$ should be real numbers while $S_{12}=-S_{21}^\ast$, so that $S_{21}/S_{12}=e^{i\phi}$, with $\phi$ some real angle.
Furthermore, the traceless property of $H(t)$ results in (iv) $\det[S(\theta_f,\theta_i)]=1$. 
Floquet theory \cite{sm} tells us that the quasienergies are just the eigenvalues of the time-averaged Hamiltonian and hence (v) $\text{Tr}[S(\theta_i\pm2\pi,\theta_i)]=2\cos(2\pi\sqrt{1-g_0^2}/\omega)$. 


The properties (i) to (v) depend only on the symmetry of the Hamiltonian and tell us important information about the transfer matrix. They could be used to check the correctness of any analytical solution or numerical simulations.  
For example, we find that numerical simulations using double precision floating numbers could give wrong results, breaking one or more properties \cite{sm}, indicating necessary high precision in obtaining the correct dynamics and an essential role played by noise. This observation makes it necessary to revisit the chiral state transfer and noise effect in a more systematical way.

{\it Exact Solution of the Transfer Matrix.---}
The encircling dynamics without noise can be solved exactly \cite{sm}: 
\begin{equation}\label{eq:exact:transferM}
  S(\theta_f,\theta_i)= \frac{\Gamma(p_1) e^{i(\theta_f-\theta_i) \sqrt{1-g_0^2}/\omega}\, p_1 \eta_i^{p_2-1}}{\Gamma(p_2) e^{(\eta_f+\eta_i)/2}}\, M_f \tilde{M}_i,
\end{equation}
where $\Gamma$ is the gamma function, $\eta_{f,i}=-2i\frac{\rho}{\omega} e^{i\theta_{f,i}}$, $M_f=M(\eta_f), \tilde{M}_i=\tilde{M}(\eta_i)$, and 
\begin{subequations}\label{eq:exact:MM}
\begin{align}
&M(\eta)=\left[
         \begin{array}{cc}
           F^{(0)}/(\omega p_1) & \quad U^{(0)}/(\omega p_1) \\
           -F^{(0)} -\frac{\eta}{p_2} F^{(1)}  & \quad - U^{(0)} +\eta  U^{(1)}  \\
         \end{array}
       \right], \label{eq:exact:M} \\
&\tilde{M}(\eta)=\left[
                    \begin{array}{cc}
                      -\omega p_1 \left(U^{(0)} -\eta  U^{(1)} \right) & \quad -U^{(0)} \\
                      \omega p_1 \left(F^{(0)} +\frac{\eta}{p_2} F^{(1)} \right) & \quad  F^{(0)} \\
                    \end{array}
                  \right], \label{eq:exact:Mtilde}
\end{align}
\end{subequations}
with $F$ and $U$ the confluent hypergeometric functions of the first and second kind, respectively, while $F^{(n)}$ and $U^{(n)}$ the abbreviations for $F(n+p_1, n+p_2,\eta)$ and $U(n+p_1, n+p_2,\eta)$, respectively \cite{PhysRevLett.118.093002}.

\begin{figure}[tbp]
  \includegraphics[width=0.48\textwidth]{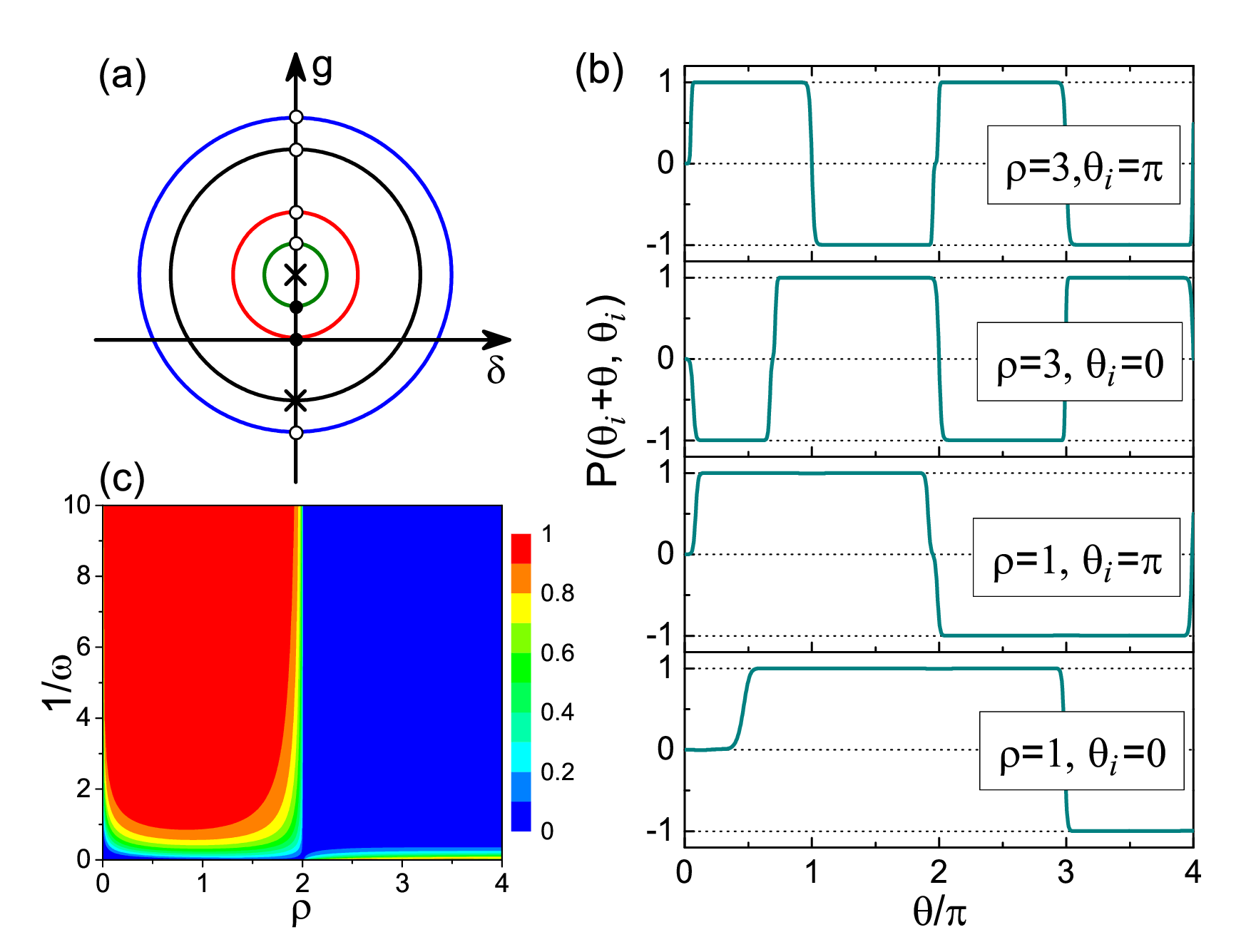} 
  \caption{(a) The loops given by Eq.(\ref{eq:loops}) with $g_0=\kappa=1$ in the complex $h_z$ plane, which may encircle one or two EPs (marked with $\times$). The solid and empty dots depict the starting/ending points in the PT-symmetric and broken phase, respectively. (b) The transition asymmetry $P(\theta_i+\theta,\theta_i)$ as a function of $\theta=\omega t$ for CCW loops, with $\omega=0.2$, $\rho=1,3$ and $\theta_i=0,\pi$. (c) The transition asymmetry after CCW one-cycle evolution starting from $\theta_i=0$, i.e., $P(2\pi,0)$, as a function of $\rho$ and $1/\omega$. } \label{fig:1} 
\end{figure}

For concreteness, we set $g_0=\kappa=1$, so that the loops encircle one (when $\rho<2$) or two (when $\rho>2$) EPs, as shown in Fig.\ref{fig:1}(a). One intriguing property of the non-Hermitian dynamics around EP is the breakdown of the adiabatic theorem, manifesting itself in a series of non-adiabatic transitions (NATs) during the evolution. Suppose the right(left) eigenvectors of $H(t)$ are $|R_{\pm}(t)\rangle$ ($|L_{\pm}(t)\rangle$), which satisfies the biorthogonal relations $\langle L_\alpha(t)|R_\beta(t)\rangle =\delta_{\alpha\beta}$. Then we can define the \emph{relative transition probability} from $|R_{\alpha}(0)\rangle$ to $|R_{\beta}(t)\rangle$ as
\begin{equation}\label{eq:transition1}
  P_{\beta\alpha}(\theta_f,\theta_i) \equiv \frac{\left|[S(\theta_f,\theta_i)]_{\beta,\alpha} \right|^2}{\left|[S(\theta_f,\theta_i)]_{+,\alpha} \right|^2 + \left|[S(\theta_f,\theta_i)]_{-,\alpha} \right|^2},
\end{equation}
where $[S(\theta_f,\theta_i)]_{\beta,\alpha}\equiv\langle L_\beta(t)|S(\theta_f,\theta_i)|R_{\alpha}(0)\rangle$. The quantity $P_{\beta\alpha}$ is symmetric for Hermitian two-level systems but asymmetric for non-Hermitian ones. The degree of asymmetry can be measured by $P(\theta_f,\theta_i)\equiv P_{+-}(\theta_f,\theta_i) -P_{-+}(\theta_f,\theta_i)$. When $P(\theta_f,\theta_i)\approx \pm1$, the corresponding eigenvector $|R_{\pm}(t)\rangle$ would dominate the state vector $|\psi(t)\rangle$. On the other hand, if $P(\theta_f,\theta_i)\approx 0$, no eigenvector can dominate the final state vector and interference effect could happen as in the Hermitian dynamics. Fig.\ref{fig:1}(b) plots this transition asymmetry as a function of the evolution time for CCW loops with the parameters $\omega=0.2$, $\rho=1,3$ and $\theta_i=0,\pi$, showing a series of sharp NATs. However, at the end of one-cycle evolution, i.e., at $\theta=2\pi$, the transition asymmetry $P\approx 1$ for $\rho=1,\theta_i=0$ but deviates from $\pm1$ and lies near an NAT in the other three cases. Fig.\ref{fig:1}(c) plots such one-cycle result $P(2\pi,0)$ in a larger $\rho$-$(1/\omega)$ space, clearly showing two different regions in the adiabatic limit: $P\approx 1$ for $\rho<2$ while $P\approx 0$ for $\rho>2$. We can expect chiral dynamics in the $P\approx 1$ region and interference behaviors in the $P\approx 0$ region, with the latter not yet discussed. 

To verify the expected interference behavior, let's analyze the $\rho\rightarrow\infty$ limit by using the asymptotic expansion of $F$ and $U$, which results in
\begin{subequations}\label{eq:asymptitic}
  \begin{align} 
  S_{12} &\rightarrow -\frac{2\pi i}{\omega}  e^{-\eta_0} \left(\frac{(-\eta_0)^{-i/\omega}}{\Gamma(1-i/\omega)}\right)^2 ,\\
  S_{21} &\rightarrow -\frac{2\pi i}{\omega}  e^{\eta_0} \left(\frac{(\eta_0)^{i/\omega}}{\Gamma(1+i/\omega)}\right)^2,
  \end{align}
\end{subequations}
for $\theta_0=0$ and hence $\eta_0=-2i\rho/\omega$. So that $S_{21}/S_{12}=e^{i\phi}$, with $\phi\approx 4(1+\rho+\log(\rho))/\omega-\pi$ \cite{sm}.  When $\phi=0 \;(\text{mod}\; 2\pi)$, $S_{12}=S_{21}$ and hence $S_{\text{CCW}}=S_{\text{CW}}$ (no chiral dynamics) according to the symmetry property (ii) given in the last section. While when $\phi=\pi \;(\text{mod}\; 2\pi)$, $S_{12}=-S_{21}$ and hence chiral dynamics happens. The angle $\phi$ varies with $\rho$ and $\omega$, so that the dynamics varies between chiral and non-chiral behaviors. This chirality oscillation is just the expected interference and has not been discovered before. We would study it quantitatively in the following.

\begin{figure}
  \centering
  \includegraphics[width=0.48\textwidth]{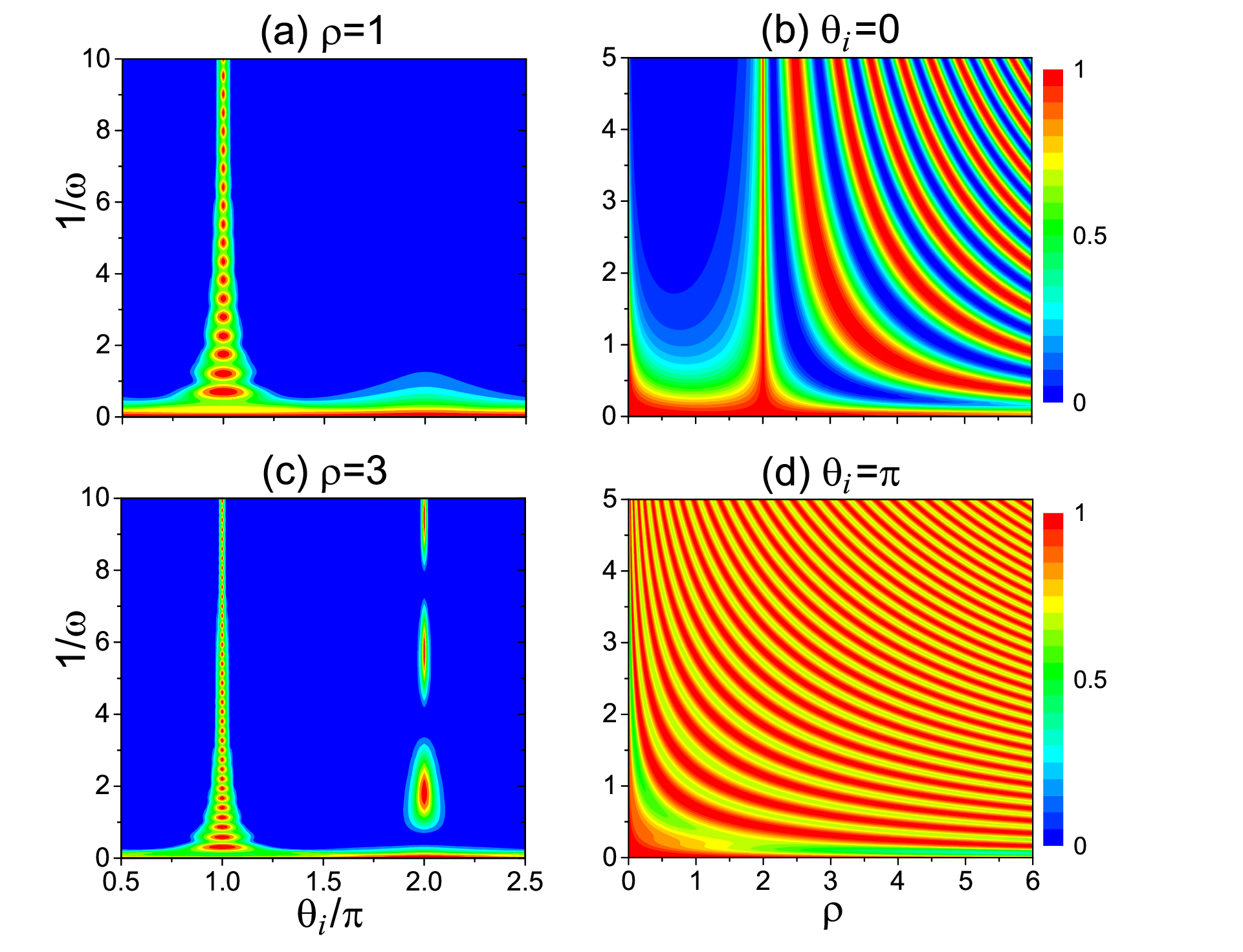}
  \caption{The non-chirality degree $\chi_c$ in the parameter space: the $(1/\omega)$-$\theta_i$ space for (a) $\rho=1$ and (c) $\rho=3$, and the $(1/\omega)$-$\rho$ space for (b) $\theta_i=0$ and (d) $\theta_i=\pi$. The results are evaluated from the exact solution Eq.(\ref{eq:exact:transferM}). } \label{fig:2}
\end{figure}
For this purpose, a quantitative measure of the chirality should be defined. We propose the following ``non-chirality degree'' $\chi_c$. In one-cycle evolution, the adiabatic states $|R_{\alpha}(0)\rangle$ evolve to $|\psi_{\text{CCW}}^{\alpha}\rangle= S_{\text{CCW}}(2\pi,0)|R_{\alpha}(0)\rangle$ and $|\psi_{\text{CW}}^{\alpha}\rangle= S_{\text{CW}}(-2\pi,0)|R_{\alpha}(0)\rangle$ in the two different directions. In the adiabatic representation, the end states could be decomposed as
$|\psi_{\text{L}}^{\alpha}\rangle =p_{\text{L}}^\alpha |R_{+}(0)\rangle +q_{\text{L}}^\alpha |R_{-}(0)\rangle$, where L denotes CCW or CW. Then the chirality can be measured in terms of the inner product between $(p_{\text{CCW}}^\alpha, q_{\text{CCW}}^\alpha)$ and $(p_{\text{CW}}^\alpha, q_{\text{CW}}^\alpha)$ with normalization:
\begin{equation}\label{eq:chirality}
  \chi_c^\alpha \equiv \frac{\left| p_{\text{CCW}}^{\alpha\ast} p_{\text{CW}}^\alpha + q_{\text{CCW}}^{\alpha\ast} q_{\text{CW}}^\alpha \right|^2 }{\left(|p_{\text{CCW}}^{\alpha}|^2+|q_{\text{CCW}}^{\alpha}|^2 \right) \left(|p_{\text{CW}}^{\alpha}|^2+|q_{\text{CW}}^{\alpha}|^2 \right)}.
\end{equation}
Obviously $\chi_c^{\alpha}=0(1)$ indicates complete chiral (non-chiral) state conversion starting from the adiabatic state $|R_\alpha(0)\rangle$. However, in general cases $\chi_c^{\alpha}$ may take all possible real values in the interval $[0,1]$ which measures the `non-chirality degree of state conversion', or simply \emph{non-chirality degree}. Furthermore, we can also remove the dependence on the initial state by defining an average $\chi_c\equiv \frac{1}{2}(\chi_c^{+}+\chi_c^{-})$. This quantity makes the chirality computable and hence enables us to study the state conversion behavior systematically and quantitatively in the whole parameter space.
In Fig.\ref{fig:2} we plot numerical results of $\chi_c$ evaluated from the exact solution given in Eq.(\ref{eq:exact:transferM}). We indeed observe the expected oscillation behavior, and such oscillations only exist when the starting point lies near the broken phase. 

One remark is needed here. To obtain the $\chi_c$ oscillation in the small-$\omega$ region in Fig.\ref{fig:2}, very high precision should be used in numerical evaluation of the exact solutions. Double precision may lead to completely wrong results, breaking one or more symmetry properties of the transfer matrix. This explains why such oscillation has not been discovered in previous theoretical studies. However, an essential question arises immediately: which result, the high-precision one or the double-precision one, should count for experimental measurements? The answer lies in the noise effect which will be demonstrated in the following.

\begin{figure}
  \centering
  \includegraphics[width=0.48\textwidth]{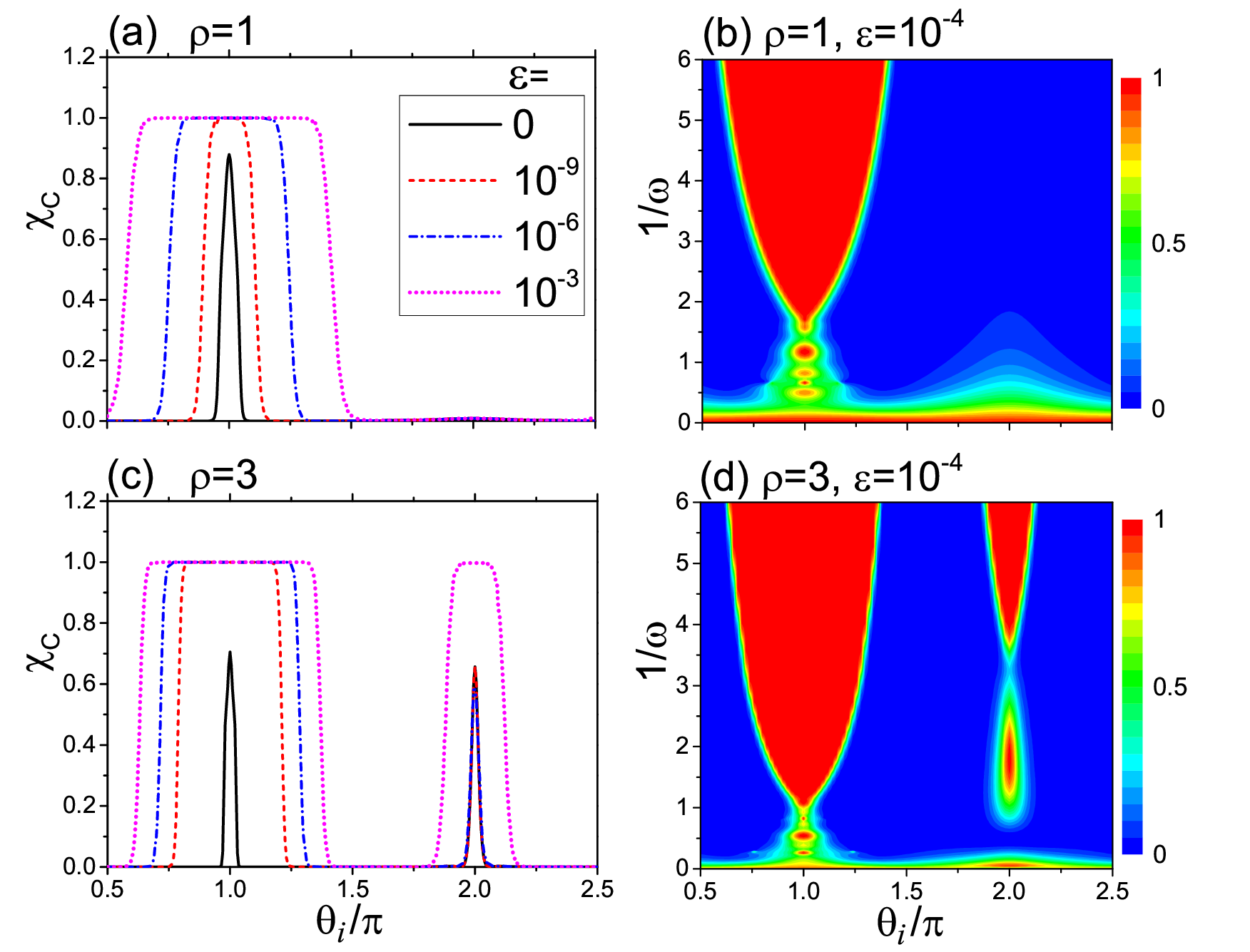}
  \caption{The non-chirality degree $\chi_c$ in the noisy dynamics under the Hamiltonian given by Eq.(\ref{eq:Hamiltonian:noisy}). (a) As a function of $\theta_i$ for $\rho=1, \omega=0.2$ and $\varepsilon=0, 10^{-9}, 10^{-6}, 10^{-3}$. (b) In the $(1/\omega)$-$\theta_i$ plane for fixed $\rho=1$ and $\varepsilon=10^{-4}$. Compare this noisy dynamics results with the corresponding results without noise shown in Fig.\ref{fig:2}(a). (c) The same as (a) except for $\rho=3$. (d) The same as (b) except for $\rho=3$.  Compare this with the corresponding results without noise shown in Fig.\ref{fig:2}(c).} \label{fig:3}
\end{figure}
{\it Speed-Noise Competition on the Chirality.---}
Noise effects are unavoidable in practical experiments and even in numerical simulations using floating numbers.
To be concrete, here we consider the dynamics under a perturbed Hamiltonian 
\begin{equation}\label{eq:Hamiltonian:noisy}
  \tilde{H}(t)=H(t)+\varepsilon \xi(t)\sigma_z, 
\end{equation}
where $\xi(t)$ is a Gaussian white noise with $\langle\xi(t)\rangle=0$ and $\langle\xi(t)\xi(t')\rangle=\delta(t-t')$. The parameter $\varepsilon$ measures the strength of the noise term. Of course, a more general noise perturbation may be written as
$\sum_{j=x,y,z}\varepsilon_j \xi_j(t)\sigma_j$. However, the $\sigma_z$ term itself is enough to reveal the qualitative behaviors and hence we consider only this term for simplicity. 

Now the dynamics could not be solved exactly and numerical methods should be applied. Here we choose the fourth-order Runge-Kutta method and use large enough significant digits to remove additional noise effects brought by roundoff errors. The resulted $\chi_c$ is plotted in Fig.\ref{fig:3}. For $\rho=1$ or $3$, the non-chirality degree $\chi_c$ is large only in a finite region near the broken phase, and this non-chiral region is enlarged as $\varepsilon$ increases [Fig.\ref{fig:3}(a) and (c)]. When $\varepsilon=10^{-4}$, the chirality oscillation near $\theta_i=\pi$ in the absence of noise [see Fig.\ref{fig:2}(a)] is erased in the small-$\omega$ region, replaced by non-chiral dynamics ($\chi_c\approx1$) [Fig.\ref{fig:3}(b)]. As $\rho>2$, the chirality oscillations near both $\theta_i=0$ and $\pi$ [Fig.\ref{fig:2}(c)] are all erased by a finite noise in the small-$\omega$ region [Fig.\ref{fig:3}(d)]. So the non-chiral dynamics starting from the broken phase is a result of noise, but not an intrinsic property as taken in some previous studies \cite{PhysRevX.8.021066}. The chirality is strongly sensitive to noise when the starting point lies in the broken phase but not sensitive to noise when the starting point lies near the symmetric phase ($\theta_i=0,\rho<2$). This observation is consistent with a recent experiment \cite{CP8_91_2025}.

\begin{figure}
  \centering
  \includegraphics[width=0.48\textwidth]{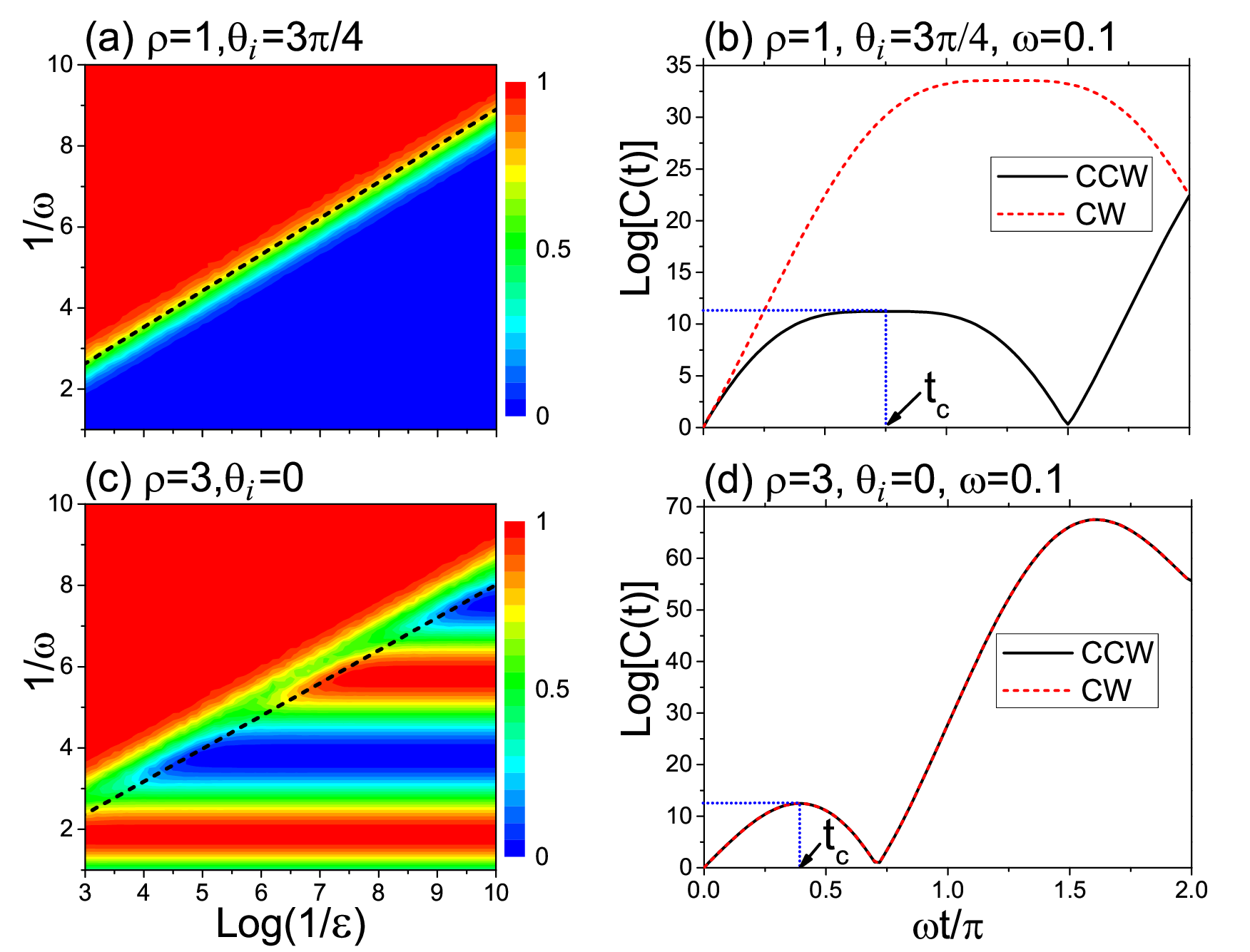}
  \caption{The non-chirality degree $\chi_c$ in the $\log(1/\varepsilon)$-$(1/\omega)$ plane [(a) and (c)] and the logarithm of the condition number $C(t)$ defined in Eq.(\ref{eq:condition}) [(b) and (d)]. The maximums of $C(t)$ at some critical time $t_c$ are depicted by blue dotted line in (b) and (d). The black dashed lines in (a) and (c) correspond to the critical lines determined by the condition $\varepsilon_c\,C(t_c)=1$.  } \label{fig:4}
\end{figure}

An important lesson learned from Fig.\ref{fig:3} is that noise could change the chirality of the state conversion process completely at some critical strength $\varepsilon_c$, which depends on $\rho, \theta_i$ and $\omega$. The dependence on the two geometric parameters $\rho$ and $\theta_i$ shows no clear universal behavior. On the other hand, the dependence on the dynamic parameter $\omega$ is quite clear: {the speed $\omega$ and the noise strength $\varepsilon$ compete with each other in determining $\chi_c$}. In Fig.\ref{fig:4}(a) and (c) we plot $\chi_c$ as a function of $\omega$ and $\varepsilon$, showing a critical line $\log(1/\varepsilon_c)\sim (1/\omega)$ which separates the parameter space into two regions with completely different dynamical behaviors. It's useful to define two limits to understand the above competition: (i) the \emph{noisy limit} with $\omega\rightarrow 0$ for fixed small but finite $\varepsilon$, and (ii) the \emph{clean limit} with $\varepsilon\rightarrow0$ for fixed small $\omega$. The critical line just separates the two limits which should be identified carefully in practical experiments and numerical simulations to remove possible misconceptions. 

We should note that although the noisy limit is defined in the $\omega\rightarrow0$ limit, it's quite different from the \emph{adiabatic limit} for Hermitian systems in the sense that the adiabatic theorem does not apply in such non-Hermitian dynamics, due to the exponential instability with respect to noise. 
We can understand this instability and hence the $\omega$-$\varepsilon$ competition in terms of perturbation theory. Up to first order in $\varepsilon$, the transfer matrix $\tilde{S}=\tilde{S}(\theta_i+2\pi,\theta_i)$ under $\tilde{H}(t)$ reads 
\begin{equation}\label{eq:perturb}
 \tilde{S} \approx S -i\varepsilon \int_{0}^{T} dt_1\; \xi(t_1) S(\theta_f,\theta_1)\sigma_z S(\theta_1,\theta_i),
\end{equation}
where $S=S(\theta_i+2\pi,\theta_i), \theta_1=\theta_i+\omega t_1$ and $T=2\pi/\omega$. So the noise term disrupts the perfect evolution at all possible intermediate times $t_1$. The relative change of the final result $||\tilde{S}-S||/||S||$ due to the noise at $t_1$ is determined by the {condition number} of $S(\theta_f,\theta_1)$. Let's define
\begin{equation}\label{eq:condition}
  C(t)\equiv\text{Cond}\left[S(\theta_i+2\pi,\theta_i+2\pi-\omega t)\right],
\end{equation}
where $\text{Cond}[A]=\sqrt{\lambda_{max}(A^\dag A)/\lambda_{min}(A^\dag A)}$ denotes the condition number of matrix $A$, with $\lambda_{max/min}(A^\dag A)$ the maximum/minimum eigenvalues of $A^\dag A$. Then if at some intermediate time $t_r$ the condition
$\varepsilon\,C(t_r)=1$ is satisfied, the result would be significantly changed by the noise near $T-t_r$, and the evolution before $T-t_r$ would be irrelevant so that the property of $\tilde{S}$ is determined solely by the evolution from $T-t_r$ to $T$. 
As $\varepsilon$ varies, the number of $t_r$ that satisfies $\varepsilon\,C(t_r)=1$ may change at some critical values of $\varepsilon_c$, and this change may (but not necessarily) lead to a qualitative change of the non-chirality degree $\chi_c$. For example, for $\rho=1, \theta_i=3\pi/4$ and $\omega=0.1$, the condition number $C(t)$ is not monotonic function of $t$ in both CCW and CW encircling loops [Fig.\ref{fig:4}(b)], and hence the number of $t_r$ changes when $1/\varepsilon$ equals to a local maximum of $C(t)$. So the transition line in the $\omega$-$\varepsilon$ plane could be given implicitly by the equation $\varepsilon_c C(t_c)=1$ with $t_c$ the first maximum point of $C(t)$ [Fig.\ref{fig:4}(b) and (d)]. Such determined critical lines are also plotted in Fig.4(a) and (c), being in good agreement with the exact numerical simulations.

Furthermore, we note that the condition number is mainly determined by the imaginary part of the dynamical phase, 
$C(t_c)\sim \exp\left[\frac{1}{\omega}\left|\int_{T-t_c}^{T} d\tau \text{Im}\lambda(\tau) \right|\right]$, with $\lambda(\tau)$ the eigenvalue of $H(\tau)$. Therefore the critical condition $\varepsilon_c C(t_c)=1$ reads $\varepsilon_c e^{{const.}/\omega}=1$, and hence the critical boundary in the $\omega$-$\varepsilon$ competition  satisfies the scaling relation $\log(1/\varepsilon_c)\sim (1/\omega)$ [Fig.\ref{fig:4}(a) and (c)].

Finally we point out that the speed-noise competition should be observable in experiments by controlling the encircling speed, the noise level or the starting point. For example, in coupled waveguide systems\cite{nature537-76,CP2-63, PhysRevLett.124.153903,nc13-2123}, the effective speed can be changed by changing the waveguide length $L$. So that one should observe a transition between chiral and non-chiral dynamics as $L$ crosses some critical value, which depends on the starting point, the effective loop radius, and the noise level. On the other hand, the speed of encircling could be varied by controlling the number of round trips in fiber-based photonic emulator, and the resulted chirality transition has already been observable in Fig.1(b) in the paper by Nasari \emph{et al.} \cite{nature605-256_2022}. We suggest that they could also change the noise level by controlling the bandpass filter and hence observe the speed-noise competition directly.  

{\it Conclusions and Discussions.---}
In conclusion, we have analyzed the non-Hermitian dynamics along loops encircling one or two EPs, focusing on the chiral state conversion behavior in one period. From symmetry and exact solution of the transfer matrix we demonstrate that in the absence of noise the chirality show oscillations when the starting point lies in the broken phase, which are very sensitive to noise when the speed is small. Noise and speed competes with each other in determining the chiral behavior, resulting in two different limits, namely the noisy limit and the clean limit, and a critical boundary between them. We also provide an explanation of this competition and the scaling relation of the critical boundary in terms of first-order perturbation theory. Our results make clear the significant role played by noise in non-Hermitian dynamics, and suggest that it should be taken into consideration in future theoretical and experimental studies carefully to avoid mistaking noise effect for intrinsic properties.   
A recent experiment \cite{PhysRevLett.134.146602} demonstrated a universal scaling of the chirality with respect to the encircling time around a Liouvillian exceptional point, which takes the same form as the one of the critical boundary discovered here. This suggests that the speed-noise (or dissipation) competition and corresponding scaling may be ubiquitous behaviors in both non-Hermitian and dissipative dynamics. 


\begin{acknowledgments}
This work has been partly supported by the Fundamental Research Funds for the Provincial Universities of Zhejiang,  Grant No.2021J014.
\end{acknowledgments}

\bibliography{EP-dynamics-bib}

@article{PhysRevE.61.929,
  title = {Repulsion of resonance states and exceptional points},
  author = {Heiss, W. D.},
  journal = {Phys. Rev. E},
  volume = {61},
  issue = {1},
  pages = {929--932},
  numpages = {0},
  year = {2000},
  month = {Jan},
  publisher = {American Physical Society},
  doi = {10.1103/PhysRevE.61.929},
  url = {https://link.aps.org/doi/10.1103/PhysRevE.61.929}
}

@article{Graefe_2008,
    author = {Graefe, E. M. and G\"unther, U. and Korsch, H. J. and Niederle, A. E.},
	title = {A non-Hermitian $PT$-symmetric Bose-Hubbard model: eigenvalue rings from unfolding higher-order exceptional points},
	journal = {J. Phys. A},
	year = {2008},
	month = {may},
	publisher = {{IOP} Publishing},
	volume = {41},
	number = {25},
	pages = {255206},
    doi = {10.1088/1751-8113/41/25/255206},
}

@article{GARRISON1988177,
title = "Complex geometrical phases for dissipative systems",
author = "Garrison, J.C. and Wright, E. M.",
journal = "Phys. Lett. A",
volume = "128",
number = "3",
pages = "177",
year = "1988",
issn = "0375-9601",
doi = "https://doi.org/10.1016/0375-9601(88)90905-X",
url = "http://www.sciencedirect.com/science/article/pii/037596018890905X",
}

@ARTICLE{Berry2004,
  author =       {Berry, M. V. },
  title =        {Physics of nonhermitian degeneracies},
  journal =      { Czech. J. Phys.},
  year =         {2004},
  volume =       {54},
  number =       {},
  pages =        {1039},
  source =       {https://doi.org/10.1023/B:CJOP.0000044002.05657.04},
}

@article{Keck_2003,
    author = {Keck, F. and Korsch, H. J. and Mossmann, S.},
	title = {Unfolding a diabolic point: a generalized crossing scenario},
	journal = {J. Phys. A},
	year = 2003,
	month = {feb},
	publisher = {{IOP} Publishing},
	volume = {36},
	number = {8},
	pages = {2125--2137},
	doi = {10.1088/0305-4470/36/8/310},
}

@article{PhysRevLett.103.093902,
  title = {Observation of $\mathcal{P}\mathcal{T}$-Symmetry Breaking in Complex Optical Potentials},
  author = {Guo, A. and Salamo, G. J. and Duchesne, D. and Morandotti, R. and Volatier-Ravat, M.
            and Aimez, V. and Siviloglou, G. A. and Christodoulides, D. N.},
  journal = {Phys. Rev. Lett.},
  volume = {103},
  issue = {9},
  pages = {093902},
  numpages = {4},
  year = {2009},
  month = {Aug},
  publisher = {American Physical Society},
  doi = {10.1103/PhysRevLett.103.093902},
  url = {https://link.aps.org/doi/10.1103/PhysRevLett.103.093902}
}

@article {Feng972,
	author = {Feng, Liang and Wong, Zi Jing and Ma, Ren-Min and Wang, Yuan and Zhang, Xiang},
	title = {Single-mode laser by parity-time symmetry breaking},
	volume = {346},
    journal = {Science},
	number = {6212},
	pages = {972--975},
	year = {2014},
	doi = {10.1126/science.1258479},
	publisher = {American Association for the Advancement of Science},
	issn = {0036-8075},
	URL = {https://science.sciencemag.org/content/346/6212/972},
}

@article {Hodaei975,
	author = {Hodaei, Hossein and Miri, Mohammad-Ali and Heinrich, Matthias and Christodoulides,
              Demetrios N. and Khajavikhan, Mercedeh},
	title = {Parity-time{\textendash}symmetric microring lasers},
    journal = {Science},
	volume = {346},
	number = {6212},
	pages = {975--978},
	year = {2014},
	doi = {10.1126/science.1258480},
	publisher = {American Association for the Advancement of Science},
    issn = {0036-8075},
	URL = {https://science.sciencemag.org/content/346/6212/975},
}

@article{PhysRevLett.106.213901,
  title = {Unidirectional Invisibility Induced by $\mathcal{P}\mathcal{T}$-Symmetric Periodic Structures},
  author = {Lin, Zin and Ramezani, Hamidreza and Eichelkraut, Toni and Kottos, Tsampikos and Cao, Hui 
           and Christodoulides, Demetrios N.},
  journal = {Phys. Rev. Lett.},
  volume = {106},
  issue = {21},
  pages = {213901},
  numpages = {4},
  year = {2011},
  month = {May},
  publisher = {American Physical Society},
  doi = {10.1103/PhysRevLett.106.213901},
  url = {https://link.aps.org/doi/10.1103/PhysRevLett.106.213901}
}

@ARTICLE{Feng2013,
  author =       {Feng, Liang and Xu, Ye-Long and William, S. Fegadolli and Lu, Ming-Hui and Jose, E. B.
                Oliveira and  Vilson, R. Almeida and  Chen, Yan-Feng and Axel, Scherer},
  title =        {Experimental demonstration of a unidirectional reflectionless parity-time
                 metamaterial at optical frequencies},
  journal =      {Nat. Mater.},
  year =         {2013},
  volume =       {12},
  number =       {},
  pages =        {108 --- 113},
  month =        {},
  note =         {},
  source =       {},
}

@article{Uzdin_2011,
doi = {10.1088/1751-8113/44/43/435302},
url = {https://doi.org/10.1088/1751-8113/44/43/435302},
year = {2011},
month = {oct},
publisher = {IOP Publishing},
volume = {44},
number = {43},
pages = {435302},
author = {Uzdin, Raam and Mailybaev, Alexei and Moiseyev, Nimrod},
title = {On the observability and asymmetry of adiabatic state flips generated by exceptional points},
journal = {J. Phys. A: Math. Theor.},
}

@article{PhysRevA.88.033842,
  title = {Breakdown of adiabatic transfer of light in waveguides in the presence of absorption},
  author = {Graefe, Eva-Maria and Mailybaev, Alexei A. and Moiseyev, Nimrod},
  journal = {Phys. Rev. A},
  volume = {88},
  issue = {3},
  pages = {033842},
  numpages = {7},
  year = {2013},
  month = {Sep},
  publisher = {American Physical Society},
  doi = {10.1103/PhysRevA.88.033842},
  url = {https://link.aps.org/doi/10.1103/PhysRevA.88.033842}
}

@article{PhysRevA.88.010102,
  title = {Time-asymmetric quantum-state-exchange mechanism},
  author = {Gilary, Ido and Mailybaev, Alexei A. and Moiseyev, Nimrod},
  journal = {Phys. Rev. A},
  volume = {88},
  issue = {1},
  pages = {010102},
  numpages = {5},
  year = {2013},
  month = {Jul},
  publisher = {American Physical Society},
  doi = {10.1103/PhysRevA.88.010102},
  url = {https://link.aps.org/doi/10.1103/PhysRevA.88.010102}
}

@ARTICLE{nature537-76,
  author =       {Doppler, J. and Mailybaev, A. A. and B\"ohm, J. and Kuhl, U. and Girschik,  A. and
                 Libisch, F. and Milburn, T. J. and Rabl, P. and  Moiseyev, N. and Rotter, S.},
  title =        {Dynamically encircling an exceptional point for asymmetric mode switching},
  journal =      {Nature},
  year =         {2016},
  volume =       {537},
  number =       {},
  pages =        {76},
  month =        {},
  note =         {},
  source =       {https://doi.org/10.1038/nature18605},
}

@ARTICLE{nature537-80,
  author =       {Xu, H. and Mason, D. and Luyao Jiang and Harris, J. G. E.},
  title =        {Topological energy transfer in an optomechanical system with exceptional points},
  journal =      {Nature},
  year =         {2016},
  volume =       {537},
  number =       {},
  pages =        {80},
  month =        {},
  note =         {},
  source =       {https://doi.org/10.1038/nature18604},
}

@article{PhysRevLett.118.093002,
  title = {Dynamically Encircling Exceptional Points: Exact Evolution and Polarization State Conversion},
  author = {Hassan, Absar U. and Zhen, Bo and Solja\ifmmode \check{c}\else \v{c}\fi{}i\ifmmode \acute{c}\else \'{c}\fi{}, Marin and Khajavikhan, Mercedeh and Christodoulides, Demetrios N.},
  journal = {Phys. Rev. Lett.},
  volume = {118},
  issue = {9},
  pages = {093002},
  numpages = {6},
  year = {2017},
  month = {Mar},
  publisher = {American Physical Society},
  doi = {10.1103/PhysRevLett.118.093002},
  url = {https://link.aps.org/doi/10.1103/PhysRevLett.118.093002}
}

@ARTICLE{nature562-86,
  author =       {Yoon, J. W. and Choi, Y. and Hahn, C. and Kim, G. and Ho Song, S. and
              Yang, K. -Y. and Yub Lee, J. and Kim, Y. and Lee, C. S. and Shin, J. K. and Lee, H. -S. and
              Berini, P.},
  title =        {Time-asymmetric loop around an exceptional point over the full optical communications band },
  journal =      {Nature},
  year =         {2018},
  volume =       {562},
  number =       {},
  pages =        {86},
  month =        {},
  note =         {},
  source =       {},
}

@article{PhysRevX.8.021066,
  title = {Dynamically Encircling Exceptional Points: In situ Control of Encircling Loops and the Role of the Starting Point},
  author = {Zhang, Xu-Lin and Wang, Shubo and Hou, Bo and Chan, C. T.},
  journal = {Phys. Rev. X},
  volume = {8},
  issue = {2},
  pages = {021066},
  numpages = {18},
  year = {2018},
  month = {Jun},
  publisher = {American Physical Society},
  doi = {10.1103/PhysRevX.8.021066},
  url = {https://link.aps.org/doi/10.1103/PhysRevX.8.021066}
}

@ARTICLE{LSA8-88,
  author =       {Zhang, X. L. and Jiang, T. and Chan, C.T. },
  title =        {Dynamically encircling an exceptional point
             in anti-parity-time symmetric systems: asymmetric mode switching for symmetry-broken modes.},
  journal =      {Light: Sci. Appl.},
  year =         {2020},
  volume =       {8},
  number =       {},
  pages =        {88},
  month =        {},
  note =         {},
  source =       {},
  doi =          {10.1038/s41377-019-0200-8},
}

@article{LSA13-65,
	title = {Chiral transmission by an open evolution trajectory in a non-{Hermitian} system},
	volume = {13},
	issn = {2047-7538},
	url = {https://doi.org/10.1038/s41377-024-01409-1},
	doi = {10.1038/s41377-024-01409-1},
    number = {1},
	journal = {Light: Sci. Appl.},
	author = {Shu, Xiaoqian and Zhong, Qi and Hong, Kai and You, Oubo and Wang, Jian and Hu, 
              Guangwei and Al\'u, Andrea and Zhang, Shuang and Christodoulides, Demetrios N. and Chen, Lin},
	month = {mar},
	year = {2024},
	pages = {65},
}

@article{LSA14-77,
	title = {Polarization-controlled chiral transport },
	volume = {14},
	issn = {2047-7538},
	url = {https://doi.org/10.1038/s41377-025-01762-9},
	doi = {10.1038/s41377-025-01762-9},
    number = {1},
	journal = {Light: Sci. Appl.},
	author = {Zhu, Hang and Wang, Jian and Al\'u, Andrea and Chen, Lin},
	month = {},
	year = {2025},
	pages = {77},
}

@ARTICLE{CP2-63,
  author =       {Zhang, X. L. and Chan, C.T. },
  title =        {Dynamically encircling exceptional points in a three-mode waveguide system},
  journal =      {Commun. Phys.},
  year =         {2019},
  volume =       {2},
  number =       {},
  pages =        {63},
  month =        {},
  note =         {},
  source =       {},
}

@article{PhysRevLett.124.153903,
  title = {Efficient Mode Transfer on a Compact Silicon Chip by Encircling Moving Exceptional Points},
  author = {Liu, Qingjie and Li, Shuyi and Wang, Bing and Ke, Shaolin and Qin, Chengzhi and Wang, Kai and Liu, Weiwei and Gao, Dingshan and Berini, Pierre and Lu, Peixiang},
  journal = {Phys. Rev. Lett.},
  volume = {124},
  issue = {15},
  pages = {153903},
  numpages = {6},
  year = {2020},
  month = {Apr},
  publisher = {American Physical Society},
  doi = {10.1103/PhysRevLett.124.153903},
  url = {https://link.aps.org/doi/10.1103/PhysRevLett.124.153903}
}

@article{PhysRevA.99.063831,
  title = {Distinct outcomes by dynamically encircling an exceptional point along homotopic loops},
  author = {Zhang, Xu-Lin and Song, Jun-Feng and Chan, C. T. and Sun, Hong-Bo},
  journal = {Phys. Rev. A},
  volume = {99},
  issue = {6},
  pages = {063831},
  numpages = {9},
  year = {2019},
  month = {Jun},
  publisher = {American Physical Society},
  doi = {10.1103/PhysRevA.99.063831},
  url = {https://link.aps.org/doi/10.1103/PhysRevA.99.063831}
}

@article{PhysRevLett.125.187403,
  title = {Hamiltonian Hopping for Efficient Chiral Mode Switching in Encircling Exceptional Points},
  author = {Li, Aodong and Dong, Jianji and Wang, Jian and Cheng, Ziwei and Ho, John S. and Zhang, Dawei and Wen, Jing and Zhang, Xu-Lin and Chan, C. T. and Al\`u, Andrea and Qiu, Cheng-Wei and Chen, Lin},
  journal = {Phys. Rev. Lett.},
  volume = {125},
  issue = {18},
  pages = {187403},
  numpages = {6},
  year = {2020},
  month = {Oct},
  publisher = {American Physical Society},
  doi = {10.1103/PhysRevLett.125.187403},
  url = {https://link.aps.org/doi/10.1103/PhysRevLett.125.187403}
}

@article{PhysRevLett.124.133905,
  title = {Robust Scattered Fields from Adiabatically Driven Targets around Exceptional Points},
  author = {Fern\'andez-Alc\'azar, Lucas J. and Li, Huanan and Ellis, Fred and Al\'u,
            Andrea and Kottos, Tsampikos},
  journal = {Phys. Rev. Lett.},
  volume = {124},
  issue = {13},
  pages = {133905},
  numpages = {6},
  year = {2020},
  month = {Apr},
  publisher = {American Physical Society},
  doi = {10.1103/PhysRevLett.124.133905},
  url = {https://link.aps.org/doi/10.1103/PhysRevLett.124.133905}
}

@article{Heiss_2012,
	doi = {10.1088/1751-8113/45/44/444016},
	url = {https://doi.org/10.1088%2F1751-8113%2F45%2F44%2F444016},
	year = {2012},
	month = {oct},
	publisher = {{IOP} Publishing},
	volume = {45},
	number = {44},
	pages = {444016},
	author = {Heiss, W. D. },
	title = {The physics of exceptional points},
	journal = {J. Phys. A}
}

@article{PhysRevA.72.014104,
  title = {Geometric phase around exceptional points},
  author = {Mailybaev, Alexei A. and Kirillov, Oleg N. and Seyranian, Alexander P.},
  journal = {Phys. Rev. A},
  volume = {72},
  issue = {1},
  pages = {014104},
  numpages = {4},
  year = {2005},
  month = {Jul},
  publisher = {American Physical Society},
  doi = {10.1103/PhysRevA.72.014104},
  url = {https://link.aps.org/doi/10.1103/PhysRevA.72.014104}
}

@article{PhysRevLett.86.787,
  title = {Experimental Observation of the Topological Structure of Exceptional Points},
  author = {Dembowski, C. and Gr\"af, H.-D. and Harney, H. L. and Heine, A. and Heiss, W. D.
            and Rehfeld, H. and Richter, A.},
  journal = {Phys. Rev. Lett.},
  volume = {86},
  issue = {5},
  pages = {787--790},
  numpages = {0},
  year = {2001},
  month = {Jan},
  publisher = {American Physical Society},
  doi = {10.1103/PhysRevLett.86.787},
  url = {https://link.aps.org/doi/10.1103/PhysRevLett.86.787}
}

@ARTICLE{nature526-554,
  author =       {Gao, T. and Estrecho, E. and Bliokh, K. Y. and Liew, T. C. H. and Fraser, M. D. and
           Brodbeck, S. and Kamp, M. and Schneider, C. and H\"ofling, S. and Yamamoto, Y. and Nori, F. and
            Kivshar, Y. S. and Truscott, A. G. and Dall, R. G. and Ostrovskaya, E. A.},
  title =        {Observation of non-Hermitian degeneracies in a chaotic exciton-polariton billiard},
  journal =      {Nature},
  year =         {2015},
  volume =       {526},
  number =       {},
  pages =        {554},
  month =        {},
  note =         {},
  source =       {},
}

@article{PhysRevLett.99.173003,
  title = {Exceptional Points in Atomic Spectra},
  author = {Cartarius, Holger and Main, J\"org and Wunner, G\"unter},
  journal = {Phys. Rev. Lett.},
  volume = {99},
  issue = {17},
  pages = {173003},
  numpages = {4},
  year = {2007},
  month = {Oct},
  publisher = {American Physical Society},
  doi = {10.1103/PhysRevLett.99.173003},
  url = {https://link.aps.org/doi/10.1103/PhysRevLett.99.173003}
}

@article{PhysRevLett.103.123003,
  title = {Resonance Coalescence in Molecular Photodissociation},
  author = {Lefebvre, R. and Atabek, O. and \ifmmode \check{S}\else \v{S}\fi{}indelka, M. and Moiseyev, N.},
  journal = {Phys. Rev. Lett.},
  volume = {103},
  issue = {12},
  pages = {123003},
  numpages = {4},
  year = {2009},
  month = {Sep},
  publisher = {American Physical Society},
  doi = {10.1103/PhysRevLett.103.123003},
  url = {https://link.aps.org/doi/10.1103/PhysRevLett.103.123003}
}

@article{PhysRevLett.118.040401,
  title = {Edge Modes, Degeneracies, and Topological Numbers in Non-Hermitian Systems},
  author = {Leykam, Daniel and Bliokh, Konstantin Y. and Huang, Chunli and Chong, Y. D. and Nori, Franco},
  journal = {Phys. Rev. Lett.},
  volume = {118},
  issue = {4},
  pages = {040401},
  numpages = {6},
  year = {2017},
  month = {Jan},
  publisher = {American Physical Society},
  doi = {10.1103/PhysRevLett.118.040401},
  url = {https://link.aps.org/doi/10.1103/PhysRevLett.118.040401}
}

@article{PhysRevLett.116.133903,
  title = {Anomalous Edge State in a Non-Hermitian Lattice},
  author = {Lee, Tony E.},
  journal = {Phys. Rev. Lett.},
  volume = {116},
  issue = {13},
  pages = {133903},
  numpages = {5},
  year = {2016},
  month = {Apr},
  publisher = {American Physical Society},
  doi = {10.1103/PhysRevLett.116.133903},
  url = {https://link.aps.org/doi/10.1103/PhysRevLett.116.133903}
}

@ARTICLE{nc13-2123,
  author =       {Shu,Xiaoqian and Li, Aodong and Hu, Guangwei and Wang,Jian and Al\'u, Andrea and Chen,Lin},
  title =        {Fast encirclement of an exceptional point for highly efficient and compact chiral mode converters},
  journal =      {Nat. Commun.},
  year =         {2022},
  volume =       {13},
  number =       {1},
  pages =        {2123},
  month =        {Apr},
  doi =          {10.1038/s41467-022-29777-5},
}

@MISC{sm,
  author =       {},
  title =        {},
  howpublished = {},
  year =         {},
  month =        {},
  note =         {},
  source =       {},
}

@article{APR12-041409_2025,
    author = {Wu, Yulin and Sun, Yuhan and Liang, Qingquan and Zhang, Hong and Xia, Lipeng and Xu, Xiaochuan and Zou, Yi},
    title = {Accelerated exceptional point encirclement in anti-parity-time symmetric systems for ultra-compact chiral mode switching},
    journal = {Appl. Phys. Rev.},
    volume = {12},
    number = {4},
    pages = {041409},
    year = {2025},
    month = {11},
    issn = {1931-9401},
    doi = {10.1063/5.0257153},
    url = {https://doi.org/10.1063/5.0257153},
    eprint = { },
}

@article{PhysRevLett.129.127401,
  title = {Riemann-Encircling Exceptional Points for Efficient Asymmetric Polarization-Locked Devices},
  author = {Li, Aodong and Chen, Weijin and Wei, Heng and Lu, Guowei and Al\`u, Andrea and Qiu, Cheng-Wei and Chen, Lin},
  journal = {Phys. Rev. Lett.},
  volume = {129},
  issue = {12},
  pages = {127401},
  numpages = {6},
  year = {2022},
  month = {Sep},
  publisher = {American Physical Society},
  doi = {10.1103/PhysRevLett.129.127401},
  url = {https://link.aps.org/doi/10.1103/PhysRevLett.129.127401}
}

@article{APL127-122201_2025,
    author = {Yao, Jiabao and L\"u, Cheng and Zhang, Jianing and Song, Jie and Tang, Shuai and Jiang, Yongyuan},
    title = {Realization of encircling the exceptional point in non-Hermitian acoustic waveguide coupler},
    journal = {Appl. Phys. Lett. },
    volume = {127},
    number = {12},
    pages = {122201},
    year = {2025},
    month = {09},
    issn = {0003-6951},
    doi = {10.1063/5.0287164},
    url = {https://doi.org/10.1063/5.0287164},
    eprint = {},
}

@article{rspa.2024.0335,
    author = {Even, Nicolas and Nennig, Benoit and Lefebvre, Gautier and Perrey-Debain, Emmanuel},
    title = {Experimental dynamical encircling of an exceptional point in coupled pendulums},
    journal = {Proc. R. Soc. A},
    volume = {481},
    number = {2316},
    pages = {20240335},
    year = {2025},
    month = {06},
    issn = {1364-5021},
    doi = {10.1098/rspa.2024.0335},
    url = {https://doi.org/10.1098/rspa.2024.0335},
    eprint = {},
}

@article{EVEN2024118239,
title = {Experimental observation of exceptional points in coupled pendulums},
journal = {J. Sound Vib. },
volume = {575},
pages = {118239},
year = {2024},
issn = {0022-460X},
doi = {https://doi.org/10.1016/j.jsv.2024.118239},
url = {https://www.sciencedirect.com/science/article/pii/S0022460X24000038},
author = {Nicolas Even and Benoit Nennig and Gautier Lefebvre and Emmanuel Perrey-Debain},
}

@article{CP3-140_2020,
  title = {Direct observation of time-asymmetric breakdown of the standard adiabaticity around an exceptional point },
  author = {Choi, Youngsun and Yoon, Jae Woong and Hong, Jong Kyun and Ryu, Yeonghwa Ryu and Song, Seok Ho},
  journal = {Commun. Phys. },
  volume = {3 },
  issue = { },
  pages = {140},
  numpages = { },
  year = {2020 },
  month = { },
  publisher = { },
  doi = {10.1038/s42005-020-00409-y},
  url = {https://doi.org/10.1038/s42005-020-00409-y}
}

@article{lpor.202100675,
author = {Liu, Weiwei and Zhang, Yicong and Deng, Zhihua and Ye, Jianghua and Wang, Kai and Wang, Bing and Gao, Dingshan and Lu, Peixiang},
title = {On-Chip Chiral Mode Switching by Encircling an Exceptional Point in an Anti-Parity-Time Symmetric System},
journal = {Laser Photonics Rev.},
volume = {16},
number = {12},
pages = {2100675},
doi = {https://doi.org/10.1002/lpor.202100675},
url = {https://onlinelibrary.wiley.com/doi/abs/10.1002/lpor.202100675},
eprint = {},
year = {2022}
}

@article{nwac259_2023,
    author = {Bai, Kai and Fang, Liang and Liu, Tian-Rui and Li, Jia-Zheng and Wan, Duanduan and Xiao, Meng},
    title = {Nonlinearity-enabled higher-order exceptional singularities with ultra-enhanced signal-to-noise ratio},
    journal = {Natl. Sci. Rev.},
    volume = {10},
    number = {7},
    pages = {nwac259},
    year = {2022},
    month = {11},
    issn = {2095-5138},
    doi = {10.1093/nsr/nwac259},
    url = {https://doi.org/10.1093/nsr/nwac259},
    eprint = {},
}

@article{PhysRevLett.123.066405,
  title = {Classification of Exceptional Points and Non-Hermitian Topological Semimetals},
  author = {Kawabata, Kohei and Bessho, Takumi and Sato, Masatoshi},
  journal = {Phys. Rev. Lett.},
  volume = {123},
  issue = {6},
  pages = {066405},
  numpages = {7},
  year = {2019},
  month = {Aug},
  publisher = {American Physical Society},
  doi = {10.1103/PhysRevLett.123.066405},
  url = {https://link.aps.org/doi/10.1103/PhysRevLett.123.066405}
}

@article{PhysRevLett.127.253901,
  title = {General Rules Governing the Dynamical Encircling of an Arbitrary Number of Exceptional Points},
  author = {Yu, Feng and Zhang, Xu-Lin and Tian, Zhen-Nan and Chen, Qi-Dai and Sun, Hong-Bo},
  journal = {Phys. Rev. Lett.},
  volume = {127},
  issue = {25},
  pages = {253901},
  numpages = {6},
  year = {2021},
  month = {Dec},
  publisher = {American Physical Society},
  doi = {10.1103/PhysRevLett.127.253901},
  url = {https://link.aps.org/doi/10.1103/PhysRevLett.127.253901}
}

@article{PhysRevLett.130.157201,
  title = {Exceptional Non-Abelian Topology in Multiband Non-Hermitian Systems},
  author = {Guo, Cui-Xian and Chen, Shu and Ding, Kun and Hu, Haiping},
  journal = {Phys. Rev. Lett.},
  volume = {130},
  issue = {15},
  pages = {157201},
  numpages = {7},
  year = {2023},
  month = {Apr},
  publisher = {American Physical Society},
  doi = {10.1103/PhysRevLett.130.157201},
  url = {https://link.aps.org/doi/10.1103/PhysRevLett.130.157201}
}

@article{CP7_109_2024,
  title = {Exceptional classifications of non-Hermitian systems },
  author = {Ryu, Jung-Wan and Han, Jae-Ho and Yi, Chang-Hwan and Park,  Moon Jip and Park, Hee Chul},
  journal = {Commun. Phys. },
  volume = {7 },
  issue = { },
  pages = {109},
  numpages = { },
  year = {2024},
  month = {},
  publisher = { },
  doi = {10.1038/s42005-024-01595-9},
  url = {https://doi.org/10.1038/s42005-024-01595-9}
}

@article{CP8_91_2025,
  title = {Dynamical topology of chiral and nonreciprocal state transfers in a non-Hermitian quantum system },
  author = {Lu, Pengfei and Liu, Yang and Lao,  Qifeng and Liu,Teng and Rao,  Xinxin and Bian,Ji and Wu, Hao
             and Zhu, Feng and  Luo, Le },
  journal = {Commun. Phys. },
  volume = {8},
  issue = { },
  pages = {91},
  numpages = { },
  year = {2025},
  month = {},
  publisher = { },
  doi = {10.1038/s42005-025-01989-3 },
  url = {https://doi.org/10.1038/s42005-025-01989-3}
}

@article{NC15_1369_2024 ,
  title = {Resolving the topology of encircling multiple exceptional points },
  author = {Guria, Chitres and Zhong, Qi and Ozdemir, Sahin Kaya and Patil, Yogesh S. S. and 
          El-Ganainy, Ramy and Harris, Jack Gwynne Emmet },
  journal = {Nat. Commun. },
  volume = {15 },
  issue = { },
  pages = {1369},
  numpages = { },
  year = {2024},
  month = {},
  publisher = { },
  doi = {10.1038/s41467-024-45530-6 },
  url = {https://doi.org/10.1038/s41467-024-45530-6}
}

@article{PhysRevA.96.052129,
  title = {Chiral state conversion without encircling an exceptional point},
  author = {Hassan, Absar U. and Galmiche, Gisela L. and Harari, Gal and LiKamWa, Patrick
          and Khajavikhan, Mercedeh and Segev, Mordechai and Christodoulides, Demetrios N.},
  journal = {Phys. Rev. A},
  volume = {96},
  issue = {5},
  pages = {052129},
  numpages = {5},
  year = {2017},
  month = {Nov},
  publisher = {American Physical Society},
  doi = {10.1103/PhysRevA.96.052129},
  url = {https://link.aps.org/doi/10.1103/PhysRevA.96.052129}
}

@article{PhysRevA.102.040201,
  title = {Encircling exceptional points as a non-Hermitian extension of rapid adiabatic passage},
  author = {Feilhauer, J. and Schumer, A. and Doppler, J. and Mailybaev, A. A. and B\"ohm, J. and Kuhl, U. and Moiseyev, N. and Rotter, S.},
  journal = {Phys. Rev. A},
  volume = {102},
  issue = {4},
  pages = {040201},
  numpages = {6},
  year = {2020},
  month = {Oct},
  publisher = {American Physical Society},
  doi = {10.1103/PhysRevA.102.040201},
  url = {https://link.aps.org/doi/10.1103/PhysRevA.102.040201}
}

@article{PhysRevA.103.023531,
  title = {On-chip experiment for chiral mode transfer without enclosing an exceptional point},
  author = {Liu, Qingjie and Liu, Jibing and Zhao, Dong and Wang, Bing},
  journal = {Phys. Rev. A},
  volume = {103},
  issue = {2},
  pages = {023531},
  numpages = {5},
  year = {2021},
  month = {Feb},
  publisher = {American Physical Society},
  doi = {10.1103/PhysRevA.103.023531},
  url = {https://link.aps.org/doi/10.1103/PhysRevA.103.023531}
}

@article{nature605-256_2022,
	title = {Observation of chiral state transfer without encircling an exceptional point},
	volume = {605},
	issn = {1476-4687},
	url = {https://doi.org/10.1038/s41586-022-04542-2},
	doi = {10.1038/s41586-022-04542-2},
    number = {7909},
	journal = {Nature},
	author = {Nasari, Hadiseh and Lopez-Galmiche, Gisela and Lopez-Aviles, Helena E. and Schumer, Alexander and Hassan, Absar U. and Zhong, Qi and Rotter, Stefan and LiKamWa, Patrick and Christodoulides, Demetrios N. and Khajavikhan, Mercedeh},
	month = may,
	year = {2022},
	pages = {256},
}

@article{PhysRevResearch.5.033053,
  title = {Universal state conversion in discrete and slowly varying non-Hermitian cyclic systems: An analytic proof and exactly solvable examples},
  author = {Nye, Nicholas S.},
  journal = {Phys. Rev. Res.},
  volume = {5},
  issue = {3},
  pages = {033053},
  numpages = {45},
  year = {2023},
  month = {Jul},
  publisher = {American Physical Society},
  doi = {10.1103/PhysRevResearch.5.033053},
  url = {https://link.aps.org/doi/10.1103/PhysRevResearch.5.033053}
}

@article{PRXQuantum.6.020328,
  title = {Topological Eigenvalue Braiding and Quantum State Transfer Near a Third-Order Exceptional Point},
  author = {Zhang, He and Liu, Tong and Xiang, Zhongcheng and Xu, Kai and Fan, Heng and Zheng, Dongning},
  journal = {PRX Quantum},
  volume = {6},
  issue = {2},
  pages = {020328},
  numpages = {17},
  year = {2025},
  month = {May},
  publisher = {American Physical Society},
  doi = {10.1103/PRXQuantum.6.020328},
  url = {https://link.aps.org/doi/10.1103/PRXQuantum.6.020328}
}

@article{nc9-4808_2018,
  title = {Winding around non-Hermitian singularities },
  author = {Zhong, Qi and Khajavikhan, Mercedeh and Christodoulides, Demetrios N. and El-Ganainy, Ramy },
  journal = {Nat. Commun.  },
  volume = {9 },
  issue = { },
  pages = {4808 },
  numpages = { },
  year = {2018},
  month = {},
  publisher = { },
  doi = {10.1038/s41467-018-07105-0},
  url = {https://doi.org/10.1038/s41467-018-07105-0}
}

@article{APLQ1-046107_2024,
    author = {Nye, Nicholas S. and Kantartzis, Nikolaos V.},
    title = {Adiabatic state conversion for (a)cyclic non-Hermitian quantum Hamiltonians of generalized functional form},
    journal = {APL Quantum},
    volume = {1},
    number = {4},
    pages = {046107},
    year = {2024},
    month = {11},
    issn = {2835-0103},
    doi = {10.1063/5.0225403},
    url = {https://doi.org/10.1063/5.0225403},
    eprint = { },
}

@article{PhysRevLett.128.160401,
  title = {Topological Quantum State Control through Exceptional-Point Proximity},
  author = {Abbasi, Maryam and Chen, Weijian and Naghiloo, Mahdi and Joglekar, Yogesh N. and Murch, Kater W.},
  journal = {Phys. Rev. Lett.},
  volume = {128},
  issue = {16},
  pages = {160401},
  numpages = {6},
  year = {2022},
  month = {Apr},
  publisher = {American Physical Society},
  doi = {10.1103/PhysRevLett.128.160401},
  url = {https://link.aps.org/doi/10.1103/PhysRevLett.128.160401}
}

@article{PhysRevLett.134.146602,
  title = {Photonic Chiral State Transfer near the Liouvillian Exceptional Point},
  author = {Gao, Huixia and Sun, Konghao and Qu, Dengke and Wang, Kunkun and Xiao, Lei and Yi, Wei and Xue, Peng},
  journal = {Phys. Rev. Lett.},
  volume = {134},
  issue = {14},
  pages = {146602},
  numpages = {7},
  year = {2025},
  month = {Apr},
  publisher = {American Physical Society},
  doi = {10.1103/PhysRevLett.134.146602},
  url = {https://link.aps.org/doi/10.1103/PhysRevLett.134.146602}
}

@article{PhysRevResearch.7.013159,
  title = {Controlling transfer and chirality of topological quantum state through dissipation in quantum walk},
  author = {Tang, Xing and Chen, Tian and Zhang, Xiangdong},
  journal = {Phys. Rev. Res.},
  volume = {7},
  issue = {1},
  pages = {013159},
  numpages = {18},
  year = {2025},
  month = {Feb},
  publisher = {American Physical Society},
  doi = {10.1103/PhysRevResearch.7.013159},
  url = {https://link.aps.org/doi/10.1103/PhysRevResearch.7.013159}
}

@MISC{arxiv250204214,
  author =       {Kumar, Parveen and Gefen, Yuval and Snizhko, Kyrylo },
  title =        {General theory of slow non-Hermitian evolution},
  howpublished = {arXiv:2502.04214},
  year =         {2025},
  month =        {},
  note =         {},
  source =       {},
}

\onecolumngrid
\clearpage

\begin{center}
\textbf{\large Supplemental Material}
\end{center}

\setcounter{equation}{0}
\setcounter{figure}{0}
\setcounter{table}{0}
\setcounter{page}{1}
\setcounter{secnumdepth}{3}
\makeatletter

\renewcommand{\thefigure}{S\arabic{figure}}
\renewcommand{\theequation}{S\arabic{equation}}
\renewcommand{\thesection}{S\Roman{section}}

\section{Hamiltonian, Transfer Matrix and Exact Solution}
We consider the Hamiltonian
\begin{equation}\label{eq:sol:H}
  H= \kappa \sigma_x +h_z(t) \sigma_z, \quad \text{where} \quad
  h_z(t)= i\left(g_0 -\rho e^{i\theta} \right),
\end{equation}
with $\theta=\omega t$. The evolution equation can be rewritten into a second order differential equation for $a(t)$,
\begin{equation}\label{eq:forA}
  \frac{d^2 a}{d\theta^2} -\left[ \frac{\rho^2}{\omega^2}e^{2i\theta} - \left(\frac{2g_0}{\omega}+i \right) \frac{\rho}{\omega} e^{i \theta} +\left(\frac{g_0^2}{\omega^2}-\frac{\kappa^2}{\omega^2} \right) \right] a=0.
\end{equation}
After the substitutions
\begin{equation}\label{eq:subst}
   \eta=-2i \frac{\rho}{\omega} e^{i\theta}, \quad a(t)=e^{it\sqrt{\kappa^2-g_0^2}} e^{-\eta/2} W(\eta),
\end{equation}
the differential equation reduces to
\begin{equation}\label{eq:W}
  \eta \frac{d^2 W}{d\eta^2} +\left[ 1+\frac{2 \sqrt{\kappa^2-g_0^2}}{\omega} -\eta\right] \frac{dW}{d\eta} -\frac{1}{\omega} \left[ ig_0+\sqrt{\kappa^2-g_0^2} \right] W =0,
\end{equation}
whose general solution reads
\begin{equation}\label{eq:sol:Ww}
  W(\eta)=c_1 F(p_1, p_2, \eta)+c_2 U(p_1,p_2, \eta),
\end{equation}
where $p_1=(ig_0+\sqrt{\kappa^2-g_0^2})/\omega$, $p_2=1+2\sqrt{\kappa^2-g_0^2}/\omega$, while $F$ and $U$ are the confluent hypergeometric functions of the first and second kind, respectively. For simplicity we shall use the following abbreviations for the hypergeometric functions:
\begin{equation*}
 F(n+p_1,n+p_2,\eta)  \leftrightarrow F^{(n)}, \quad U(n+p_1,n+p_2,\eta) \leftrightarrow U^{(n)}.
\end{equation*}
We would also denote the \textbf{adjoint matrix} of a square matrix $A$ as $A^\ddag$.

From the equation of motion we have
\begin{eqnarray*}
  \kappa b &=& i\partial_t a -ig_0 a +i\rho e^{i\theta} a \\
           &=& -\omega p_1 a-\omega \eta e^{it\sqrt{\kappa^2-g_0^2}} e^{-\eta/2} \left(c_1 \frac{p_1}{p_2} F^{(1)} -c_2 p_1 U^{(1)} \right).
\end{eqnarray*}
Then $(a,b)^T= e^{it\sqrt{\kappa^2-g_0^2}} e^{-\eta/2} M(\eta) (c_1,c_2)^T$, where
\begin{equation}\label{eq:loop-A:matrixM}
  M(\eta)=\left[
         \begin{array}{cc}
           F^{(0)} & \quad U^{(0)} \\
           -\frac{\omega}{\kappa} p_1 \left(F^{(0)} +\frac{\eta}{p_2} F^{(1)} \right) & \quad -\frac{\omega}{\kappa} p_1 \left(U^{(0)} -\eta  U^{(1)} \right) \\
         \end{array}
       \right].
\end{equation}
The transfer matrix can be expressed as
\begin{equation}
  S(\theta_f,\theta_i)= e^{i(\theta_f-\theta_i)\sqrt{\kappa^2-\rho^2}/\omega} e^{-(\eta_f-\eta_i)/2} M(\eta_f) \left[M(\eta_i) \right]^{-1},
\end{equation}
where $\eta_{f,i}=-2i \frac{\rho}{\omega} e^{i\theta_{f,i}}$. 
The inverse of $M(\eta)$ can be expressed in terms of its determinant and adjoint matrix, $[M(\eta)]^{-1}=\frac{1}{\det{M(\eta)}} M(\eta)^{\ddag}$, where
\begin{eqnarray}
  \det{M(\eta)}   &=& \frac{\omega}{\kappa} \frac{\Gamma(p_2)}{\Gamma(p_1)} \eta^{1-p_2} e^{\eta}, \\
  M(\eta)^{\ddag} &=& \left[
                    \begin{array}{cc}
                      -\frac{\omega}{\kappa} p_1 \left(U^{(0)} -\eta  U^{(1)} \right) & \quad -U^{(0)}\\
                      \frac{\omega}{\kappa} p_1 \left(F^{(0)} +\frac{\eta}{p_2} F^{(1)} \right) & \quad  F^{(0)} \\
                    \end{array}
                  \right].
\end{eqnarray}
Then the transfer matrix reads
\begin{equation}
  S(\theta_f,\theta_i)= \frac{\kappa}{\omega} \frac{\Gamma(p_1)}{\Gamma(p_2)} \eta_i^{p_2-1} e^{i(\theta_f-\theta_i) \sqrt{\kappa^2-\rho^2}/\omega} e^{-(\eta+\eta_0)/2} M(\eta_f) M(\eta_i)^{\ddag}.
\end{equation}
By redefining $M$, we can obtain the solution Eq.(7) in the main text.

\section{Symmetry Analysis}
The parameter $h_z(\theta)=i(g_0-\rho e^{i\theta}),\; \theta=\omega t$. Then the Hamiltonian could also be written as a function of $\theta$:
\begin{equation*}
  H(\theta)\equiv \kappa \sigma_x +h_z(\theta)\sigma_z.
\end{equation*}
The parameter $h_z(\theta)$ has the symmetry
\begin{equation*}
  h_z(\theta)^\ast=-h_z(-\theta)=-h_z(2\pi-\theta),
\end{equation*}
and hence
\begin{equation}
  [H(\theta)]^T =H(\theta), \qquad \sigma_z H(\theta)^\ast \sigma_z = -H(2\pi-\theta).
\end{equation}
For \emph{CCW cycling}, $\theta$ increases, and the transfer matrix
\begin{equation*}
  S_{\text{CCW}}(\theta_f,\theta_i) \equiv \mathcal{T}_\theta \exp\left[-i\frac{1}{\omega}\int_{\theta_i}^{\theta_f} H(\theta)d\theta \right],
\end{equation*}
with $\theta_f>\theta_i$, and $\mathcal{T}_\theta$ means $\theta$-ordering. On the other hand, for CW cycling, $\theta$ decreases, and the corresponding transfer matrix 
\begin{equation*}
 S_{\text{CW}}(\theta_f,\theta_i) \equiv \mathcal{T}_{-\theta} \exp\left[-i\frac{1}{\omega}\int_{\theta_f}^{\theta_i} H(\theta)d\theta \right],
\end{equation*}
with $\theta_i>\theta_f$, and $\mathcal{T}_{-\theta}$ means ``inverse-$\theta$-ordering''. 
Following the symmetry of $H(\theta)$, we have
\begin{align}
  \sigma_z\left[S_{\text{CCW}}(\theta_f,\theta_i)\right]^\ast\sigma_z &=\mathcal{T}_\theta e^{-i\frac{1}{\omega} \int_{\theta_i}^{\theta_f} H(2\pi-\theta)d\theta}  \nonumber \\
   &=\mathcal{T}_{-\theta} e^{-i\frac{1}{\omega} \int_{2\pi-\theta_f}^{2\pi-\theta_i} H(\theta)d\theta} \nonumber\\
   &=S_{\text{CW}}(2\pi-\theta_f, 2\pi-\theta_i) \nonumber\\
   &=S_{\text{CW}}(-\theta_f, -\theta_i).
\end{align}
This is just the property (i) in the main text. Similarly, we can prove
\begin{equation}
  \left[S_{\text{CCW}}(\theta_f,\theta_i)\right]^T =\mathcal{T}_{-\theta}\exp\left[ -i\int_{\theta_i}^{\theta_f} H(\theta)d\theta\right]   =S_{\text{CW}}(\theta_i, \theta_f).
\end{equation}
This is the property (ii) in the main text. Combining the above two properties, we have 
\begin{equation}
  \sigma_z \left[S_{\text{CCW/CW}}(\theta_f,\theta_i)\right]^\dag \sigma_z =S_{\text{CCW/CW}}(-\theta_i,-\theta_f).
\end{equation}
This is the property (iii) in the main text. 

Furthermore, the determinant 
\begin{equation}
  \det[S_{\text{CCW}}(\theta_f,\theta_i)]=\exp\left[-i\frac{1}{\omega}\int_{\theta_i}^{\theta_f}\text{Tr}[H(\theta)] d\theta\right] =1.
\end{equation}
This is just the property (iv) in the main text.

In Floquet theory, the quasienergies $\epsilon_{1,2}$ and corresponding Floquet modes $|\Phi_{1,2}\rangle$ are defined through the eigenvalues of the one-period transfer matrix by
\begin{equation*}
  S(\theta_i+2\pi, \theta_i) |\Phi_\alpha\rangle =e^{-i(2\pi/\omega)\epsilon_\alpha}  |\Phi_\alpha\rangle. 
\end{equation*}
So that $\text{Tr}[ S(\theta_i+2\pi, \theta_i)]=e^{-i(2\pi/\omega)\epsilon_1} +e^{-i(2\pi/\omega)\epsilon_2}$. The property (iv) tells us that $\epsilon_1+\epsilon_2=0$, and hence
\begin{equation}
  \text{Tr}[ S(\theta_i+2\pi, \theta_i)]=2\cos[(2\pi/\omega)\epsilon_1].
\end{equation}
The Floquet analysis below gives that $\epsilon_{1,2}=\pm\sqrt{\kappa^2-g_0^2}$, and hence we obtain property (v) in the main text.

\section{Lessons from Numerical Simulation}
Let's focus on the numerical simulations in some special cases and learn some important lessons.

(1). For $\theta_i=\pi, \kappa/\omega=10, \rho/\kappa=1$, the starting point lies in the broken phase and the transfer matrix elements read 
      \begin{align*}
        S_{11} &= 1.6003286929485..., \\
        S_{22} &= 0.39967130705149..., \\
        S_{12} &= -0.26983259101832... + 0.53626944011531...\, i,\\
        S_{21} &= 0.26983259101832... + 0.53626944011531...\,i.
      \end{align*} 
     For $\theta_i=\pi, \kappa/\omega=10, \rho/\kappa=6$, 
      \begin{align*}
        S_{11} &= 2.0644847016278..., \\
        S_{22} &= -0.064484701627799..., \\
        S_{12} &= -0.93027858578974...-0.51740644837560...\, i,\\
        S_{21} &= 0.93027858578974... - 0.51740644837560...\, i.
      \end{align*} 
However, to obtain the above exact results, high precision should be used in numerical computation. For example, the result for $\rho/\kappa=6$ can be obtained by using precision larger than 48. Using double precision to evaluate the exact solution would result in 
\begin{equation*}
S_{\text{exact,double-precision}}\approx \left(
                   \begin{array}{cc}
                     1.06448 & -0.930279 - 0.517406 i \\
                     0.930279 - 0.517406 i & -1.06448 \\
                   \end{array}
                 \right)
\end{equation*}
It's obviously wrong since its trace is zero ( which should be 2) and its determinant is very close to zero (which should be 1). 

Even worse is to evaluate the transfer matrix by numerical integration of the evolution equation. For example, for $\theta_i=\pi, \kappa/\omega=10, \rho/\kappa=6$, using double precision, 4th-order Runge-Kutta method and taking 2000 steps in one period, the transfer matrix reads
      \begin{equation*}
        S_{\text{num,prec=16,steps=2000}}\approx 
        \left(
          \begin{array}{cc}
            -20.7508-13.5187i & -0.473215 -0.265766i \\
            -0.979836 + 1.48367 i & -0.019292 + 0.0338536 i \\
          \end{array}
        \right)\times 10^{70}.
      \end{equation*}
Now the trace, the determinant and even the symmetry are all wrong.  If we increase the working precision to be 220, the result reads
\begin{align*}
        & S_{\text{num,prec=220,steps=2000}} \\
        &\approx 
        \left(
          \begin{array}{cc}
           1.93979483799   &  -0.870121625323 + 0.355109412225 i \\
           0.87012162532 + 0.35510941223 i   & 0.060205162014  \\
          \end{array}
        \right).
      \end{align*}
Further increase the precision to 250 and the step number to 10000, 
      \begin{align*}
       & S_{\text{num,prec=250,steps=10000}} \\
       & \approx 
        \left(
          \begin{array}{cc}
           2.05973794035   & -0.942704950489 - 0.484099037946 i  \\
           0.942704950489 - 0.484099037946 i   & -0.0597379403497 \\
          \end{array}
        \right).
      \end{align*}
Now, the trace, determinant and the symmetry are all correct upto a small relative error. 

\textbf{The lesson}: In numerical simulation of non-Hermitian dynamics the \textbf{precision} and \textbf{step size} are crucial in obtaining the correct result. The reason is that \textbf{non-Hermitian dynamics is usually exponentially sensitive to noise} while the numerical roundoff error plays a role of some quenched noise.

(2). For $\theta_i=0, \kappa/\omega=10, \rho/\kappa=1$, the starting point lies in the symmetric phase and the exact transfer matrix elements are
      \begin{align*}
        S_{11} &= 1.0770577614172...\times 10^{22} \\
        S_{22} &= 1.0770577614172...\times 10^{22} \\
        S_{12} &= 1.07582386952278...\times10^{22}-5.154050052286...\times10^{20} i\\
        S_{21} &= -1.07582386952278...\times10^{22}-5.154050052286...\times10^{20} i.
      \end{align*}
Numerical integration method gives
      \begin{align*}
       & S_{\text{num,double precision,steps=10000}} \\
       & \approx 
        \left(
          \begin{array}{cc}
           1.07674\times10^{22} + 1.71295\times10^7 i    &  1.07551\times10^{22}-5.15253\times10^{20}i \\
           -1.07551\times10^{22}-5.15253\times10^{20} i  &  -1.07674\times10^{22} + 2.17907\times10^{7}i \\
          \end{array}
        \right).
      \end{align*}
The relative error is small although the absolute error is large. Increase the working precision to 100: 
       \begin{align*}
       & S_{\text{num,precision=100,steps=10000}} \\
       & \approx 
        \left(
          \begin{array}{cc}
           1.07674\times10^{22}    &  1.07551\times10^{22}-5.15253\times10^{20}i \\
           -1.07551\times10^{22}-5.15253\times10^{20}i  &  -1.07674\times10^{22}  \\
          \end{array}
        \right).
      \end{align*}
The relative error becomes even smaller. 

\textbf{The lession}: The sensitivity to noise of the non-Hermitian dynamics is dependent on the starting/ending points of the loop. In the special case studied above, the sensitivity is high (low) if the starting point lies in the broken (symmetric) phase.

\section{Floquet Analysis}
According to Floquet theory, there exist Floquet states $|\psi_\lambda(t)\rangle $ that satisfy
\begin{equation}\label{eq:Floquet}
  |\psi_\lambda (t) \rangle =e^{-i\lambda t} |\Phi_\lambda(t)\rangle, \quad \text{with}\quad |\Phi_\lambda(t+T)\rangle=|\Phi_\lambda(t)\rangle,
\end{equation}
and $\lambda$ is called \emph{quasienergy}. The quasienergy spectrum is important for stroboscopic dynamics. In this section we analyze this spectrum. It's obvious that the Floquet states should be eigenstates of the evolution operator (or transfer matrix) in one period $S\equiv S(\theta_i+2\pi,\theta_i)$:
\begin{equation}\label{eq:Floquet:eigen}
  S|\Phi_\lambda\rangle = e^{-i\lambda T} |\Phi_\lambda\rangle, \quad \text{where} \quad |\Phi_\lambda\rangle\equiv |\Phi_\lambda(0)\rangle.
\end{equation}
Hence the quasienergy spectrum can be obtained from the eigenvalues of the transfer matrix $S$.
Furthermore, since $H$ is traceless, Liouville's formula implies that $\det S=\exp\left[ -i\int_{0}^T dt \; \text{Tr} H(t) \right]=1$. This means that the two quasienergies $\lambda_{1,2}$ satisfy $\lambda_1+\lambda_2=0$ and hence we can write them as $\lambda_{1,2}=\pm\lambda$. Consequently, the transfer matrix in one period should take the form
\begin{equation*}
  S=\cos\lambda\; \mathds{I} +\vec{n}\cdot\vec{\sigma}, 
\end{equation*}
with $\vec{n}\cdot\vec{n}=0$ since $\det S=1$.

We now give a lemma about the quasienergy spectrum.
\begin{lemma}
  If the Hamiltonian takes the form $H(t)=\sum_{n=0}^{+\infty}H^{(n)} e^{i n \omega t}$, then the quasienergies $\lambda_\alpha$ are just the eigenvalues of $H^{(0)}$. 
\end{lemma}
\begin{proof}
  We prove it by using the Floquet-Hilbert formalism. 
  According to the Floquet theorem, the Schr\"odinger equation has some steady-state solutions of the form
  $|\psi_\alpha(t)\rangle=e^{-i\lambda_\alpha t}|\phi_\alpha(t)\rangle$, where $\lambda_\alpha$ are called \emph{quasienergies} while the periodic \emph{Floquet modes} $|\phi_\alpha(t)\rangle=|\phi_\alpha(t+T)\rangle$ satisfies
  \begin{equation}\label{eq:Floquet:mode}
    \left[H(t) -i\frac{\partial}{\partial t}\right]|\phi_\alpha(t)\rangle =\lambda_\alpha |\phi_\alpha(t)\rangle. 
  \end{equation}
  This equation could be seen as an steady-state equation in the Floquet-Hilbert space $\mathcal{F}=\mathcal{H}\otimes\mathcal{T}$, where $\mathcal{T}$ is the space of bounded periodic functions over $[0,T]$. The quasienergies $\lambda_\alpha$ is then the eigenvalues of the operator $\hat{Q}\equiv H(t) -i\frac{\partial}{\partial t}$. 
  A complete set of basis of $\mathcal{F}$ can be obtained by direct product between a complete set of orthonormal basis $|\alpha\rangle$ of $\mathcal{H}$ and the complete set of time-periodic functions $e^{im\omega t}$ labeled by an integer $m$, 
  \begin{equation*}
    |\alpha m(t)\rangle \equiv \langle t|\alpha m\rangle\rangle =|\alpha\rangle e^{im\omega t}. 
  \end{equation*}
  In the basis $|\alpha m\rangle\rangle$, the operator $\hat{Q}$ has a block structure, with each block representing an operator in the Hilbert space, 
  \begin{equation*}
    \hat{Q}_{mm'} = m\omega \delta_{mm'} + H^{(m-m')}. 
  \end{equation*}
  Obviously, for  $H(t)=\sum_{n=0}^{+\infty}H^{(n)} e^{i n \omega t}$, the matrix for $\hat{Q}$ is a lower triangular matrix with respect to the index $m$. So the eigenvalues should be the eigenvalues of the diagonal blocks, $m\omega +H^{(0)}$. Since $m\omega+\lambda_k$ is equivalent to $\lambda_k$ as a quasienergy, so we conclude that the quasienergies are just the eigenvalues of $H^{(0)}$.  
  
  Furthermore, the corresponding Floquet mode could also be constructed from the eigenvector of $H^{(0)}$. For this purpose, we come back to the Hilbert space and Eq.(\ref{eq:Floquet:mode}). Taking the ansatz 
  \begin{equation*}
    |\phi_\alpha(t)\rangle=\sum_{n=0}^\infty |\phi_{\alpha,n}\rangle e^{in\omega t},
  \end{equation*}
  and inserting this form into Eq.(\ref{eq:Floquet:mode}), we obtain
  \begin{equation*}
    \sum_{m'=0}^{\infty} H^{(m-m')}|\phi_{\alpha,m'}\rangle +m\omega|\phi_{\alpha,m}\rangle =\lambda_\alpha |\phi_{\alpha,m}\rangle. 
  \end{equation*}
  Since $H^{(n)}$ is nonzero only for nonnegative $n$, the above equation reduces to the following recursive relations:
  \begin{align*}
   & H^{(0)} |\phi_{\alpha,0}\rangle =\lambda_\alpha|\phi_{\alpha,0}\rangle, \\
   & \left[H^{(0)}+\omega -\lambda_\alpha\right] |\phi_{\alpha,1}\rangle =-H^{(1)}|\phi_{\alpha,0}\rangle, \\
   & \cdots \quad \cdots \\
   & \left[H^{(0)}+m\omega -\lambda_\alpha\right] |\phi_{\alpha,m}\rangle = -\sum_{m'=1}^{m} H^{(m')} |\phi_{\alpha,m-m'}\rangle,\\
   &\cdots \quad \cdots
  \end{align*}
  In this way, the Floquet mode could be constructed from the eigenvector $|\phi_{\alpha,0}\rangle$ corresponding to the eigenvalue $\lambda_\alpha$ of $H^{(0)}$ recursively. Specially, if $H^{(n)}$ does not vanish only for $n=0$ and $n=1$, then the above recursive relations can be solved explicitly, resulting in 
  \begin{equation*}
     |\phi_\alpha(t)\rangle=\left\{ \mathds{I} +\sum_{n=1}^\infty  e^{in\omega t} \prod_{m=n}^{1}\left[\left(\lambda_\alpha-m\omega-H^{(0)}\right)^{-1} H^{(1)}\right] \right\} |\phi_{\alpha,0}\rangle. 
  \end{equation*}
\end{proof}

\section{Asymptotic Analysis}
\textbf{(1) Asymptotic behavior of $F$}. 
For fixed values of $a,b$ and as $|z|\rightarrow\infty$,  $F(a,b,z)={}_1F_{1}(a,b,z)$,
\begin{align*}
  F(a,b,z) &= \frac{\Gamma(b)}{\Gamma(b-a)} (-z)^{-a} \sum_{r=0}^N \frac{(a)_r(a-b+1)_r}{r!} (-z)^{-r} +O(|z|^{-N-a-1}) \\
   &\quad + \frac{\Gamma(b)}{\Gamma(a)} e^z z^{a-b}\sum_{r=0}^M \frac{(1-a)_r(b-a)_r}{r!} (z)^{-r} +O(|e^zz^{-N+a-b-1}|).
\end{align*}
If $\text{Re}z\rightarrow\infty$, then
\begin{equation*}
 F(a,b,z) =  \frac{\Gamma(b)}{\Gamma(a)} e^z z^{a-b} [1+O(|z|^{-1})].
\end{equation*}
If $\text{Re}z\rightarrow -\infty$, then
\begin{equation*}
 F(a,b,z) =  \frac{\Gamma(b)}{\Gamma(b-a)} (-z)^{-a} [1+O(|z|^{-1})].
\end{equation*}
As $|z|\rightarrow\infty$, $-\frac{3}{2}\pi<\arg z<\frac{3}{2}\pi$, 
\begin{equation*}
  U(a,b,z)=z^{-a}  \sum_{r=0}^N (-1)^r\frac{(a)_r(a-b+1)_r}{r!} (z)^{-r}  +O(|z|^{-N-a-1}).
\end{equation*}

\textbf{(2) Loops starting from $\theta_i=\pi$.} 
As $\rho\rightarrow \infty$, the asymptotic behavior of $F(p_1,p_2,\eta_0)$ should be (note that $\eta_0=2i\rho/\omega$ is a purely imaginary number,i.e., $\text{Re}\eta_0=0$)
\begin{equation*}
 F(p_1,p_2,\eta_0) \rightarrow \frac{\Gamma(p_2)}{\Gamma(p_2-p_1)}(-\eta_0)^{-p_1} +\frac{\Gamma(p_2)}{\Gamma(p_1)} e^{\eta_0} \eta_0^{p_1-p_2}.
\end{equation*}
And hence
\begin{align*}
  F^{(0)}_0 &\rightarrow \frac{(-\eta_0)^{-i\kappa/\omega}}{\Gamma(1-i\kappa/\omega)},\\
  F^{(1)}_0 &\rightarrow \frac{(-\eta_0)^{-1-i\kappa/\omega}}{\Gamma(1-i\kappa/\omega)} +\frac{(\eta_0)^{-1+i\kappa/\omega}}{\Gamma(1+i\kappa/\omega)} e^{\eta_0}.
\end{align*}
So that
\begin{align*}
  S_{11} &= 1+\frac{2\pi\kappa}{\omega} e^{-\eta_0} F^{(0)}_0 \left(F^{(0)}_0+\eta_0F^{(1)}_0 \right) \\
  &\rightarrow 1+\frac{2\pi\kappa}{\omega}\left| \frac{(i\kappa/\omega)^{i\kappa/\omega}}{\Gamma(1+i\kappa/\omega)} \right|^2,\\
  S_{22} &=2-S_{11}\rightarrow  1-\frac{2\pi\kappa}{\omega}\left| \frac{(i\kappa/\omega)^{i\kappa/\omega}}{\Gamma(1+i\kappa/\omega)} \right|^2,\\
  S_{12} &= -\frac{2\pi i\kappa}{\omega} e^{-\eta_0} \left(F^{(0)}_0  \right)^2,\\
  &\rightarrow -\frac{2\pi i\kappa}{\omega}  e^{-\eta_0} \left(\frac{(-\eta_0)^{-i\kappa/\omega}}{\Gamma(1-i\kappa/\omega)}\right)^2 ,\\
  S_{21} &= -\frac{2\pi i\kappa}{\omega} e^{-\eta_0} \left(F^{(0)}_0+\eta_0F^{(1)}_0 \right)^2 \\
  &\rightarrow -\frac{2\pi i\kappa}{\omega}  e^{\eta_0} \left(\frac{(\eta_0)^{i\kappa/\omega}}{\Gamma(1+i\kappa/\omega)}\right)^2.
\end{align*}
From this asymptotic behavior, we can conclude that
\begin{itemize}
  \item The diagonal elements are real and the off-diagonal elements satisfy the relation $S_{21}=-S_{12}^\ast$. This is consistent with our general arguments based on the symmetry of the Hamiltonian.  [See the main text, above Eq.(4)]
  \item The ratio $\frac{S_{21}}{S_{12}}$ has the asymptotic expression
  \begin{equation*}
    \frac{S_{21}}{S_{12}}=e^{i\phi} \rightarrow e^{2\eta_0} (-\eta_0^2)^{2i\kappa/\omega} \left( \frac{\Gamma(1-i\kappa/\omega)}{\Gamma(1+i\kappa/\omega)} \right)^2,
  \end{equation*}
  and hence 
  \begin{equation*}
    \phi\approx \frac{4\kappa}{\omega}\left(1+\frac{\rho}{\kappa} +\log(\rho/\kappa)\right) -\pi.
  \end{equation*}
  The important thing is that the ratio $\frac{S_{21}}{S_{12}}$ rotates on the unit circle as $\rho/\kappa$ or $\kappa/\omega$ increases, resulting in an oscillating $\chi_c$. 
\end{itemize}

\textbf{(3) Loops starting from $\theta_i=0$.} Similar asymptotic analysis shows that
\begin{align*}
  S_{11} &\rightarrow 1+\frac{2\pi\kappa}{\omega}\left| \frac{(i\kappa/\omega)^{-i\kappa/\omega}}{\Gamma(1+i\kappa/\omega)} \right|^2,\\
  S_{22} &\rightarrow  1-\frac{2\pi\kappa}{\omega}\left| \frac{(i\kappa/\omega)^{-i\kappa/\omega}}{\Gamma(1+i\kappa/\omega)} \right|^2,\\
  S_{12} &\rightarrow -\frac{2\pi i\kappa}{\omega}  e^{-\eta_0} \left(\frac{(-\eta_0)^{-i\kappa/\omega}}{\Gamma(1-i\kappa/\omega)}\right)^2 ,\\
  S_{21} &\rightarrow -\frac{2\pi i\kappa}{\omega}  e^{\eta_0} \left(\frac{(\eta_0)^{i\kappa/\omega}}{\Gamma(1+i\kappa/\omega)}\right)^2.
\end{align*}
Still, $\frac{S_{21}}{S_{12}}=e^{i\phi}$ and $\phi$ increases monotonically as $\rho/\kappa$ or $\kappa/\omega$ increases, resulting in oscillating $\chi_c$.

\begin{figure}
  \centering
  \includegraphics[width=0.6\textwidth]{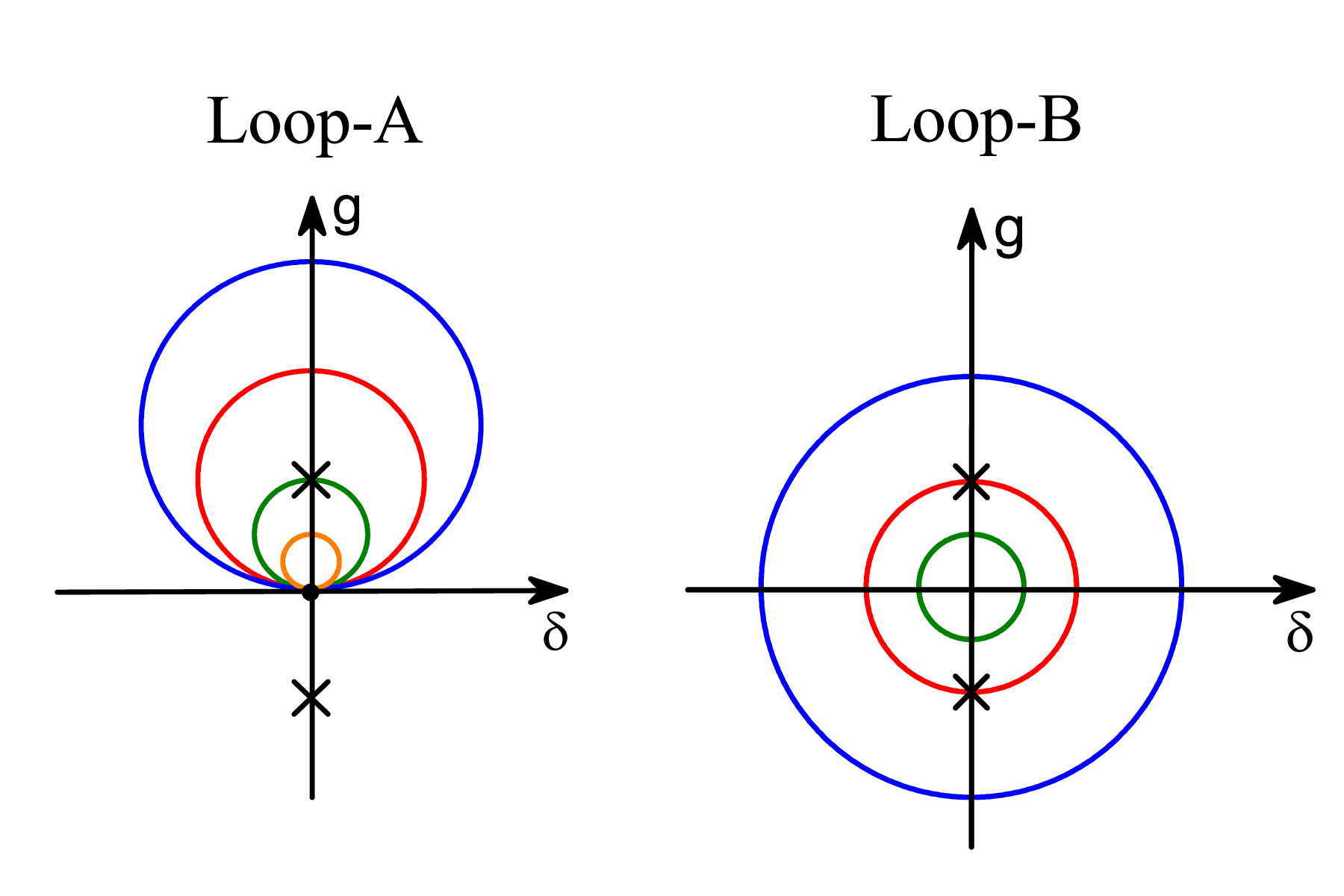}
  \caption{Two more series of loops in the complex $h_z$ plane. Exceptional points are denoted by $\times$. }\label{fig:s1}
\end{figure}

\section{Another Series of Loops: Loop-A}
In addition to the series of loops considered in the main text, we investigate two more series of loops. The first series (Loop-A) is given by setting $g_0=\rho$, i.e., $h_z=i\rho(1-e^{i\omega t})$. 
The second series (Loop-B) is given by setting $g_0=0$, and hence $h_z=-i\rho e^{i\omega t}$. See Fig.\ref{fig:s1} for the geometry of the loops in the complex $h_z$ plane.

\subsection{The Transfer Matrix, the Floquet States}
For Loop-A, we would focus on the $\kappa^2>\rho^2$ case while the $\kappa^2<\rho^2$ case can be obtained similarly.
From the general exact solution given in the main text [Eq.(7)], the transfer matrix reads
\begin{equation}
  S(\theta_f,\theta_i)= \frac{\Gamma(p_1) e^{i(\theta_f-\theta_i) \sqrt{1-\rho^2}/\omega}\, p_1 \eta_i^{p_2-1}}{\Gamma(p_2) e^{(\eta_f+\eta_i)/2}}\, M_f \tilde{M}_i,
\end{equation}
with $p_1=(i\rho+\sqrt{\kappa^2-\rho^2})/\omega$ and $p_2=1+2\sqrt{\kappa^2-\rho^2}/\omega$. 
In one-period, we should note the connection formula for $U^{(n)}$:
\begin{align*}
  U^{(0)} &\rightarrow U^{(0)}e^{-2iT\sqrt{\kappa^2-\rho^2}} -\frac{2\pi i e^{-iT\sqrt{\kappa^2-\rho^2}}}{\Gamma(p_2)\Gamma(1+p_1-p_2)} F^{(0)},\\
  U^{(1)} &\rightarrow U^{(1)} e^{-2iT\sqrt{\kappa^2-\rho^2}} +\frac{2\pi i e^{-iT\sqrt{\kappa^2-\rho^2}}}{p_2\Gamma(p_2)\Gamma(1+p_1-p_2)} F^{(1)}.
\end{align*}

The quasienergies are just given by
\begin{equation}
  \epsilon_{1,2}=\pm\sqrt{\kappa^2-\rho^2}. 
\end{equation}
The corresponding Floquet states could be derived from the exact solution, 
\begin{equation*}
  \left[
     \begin{array}{cc}
       a(t)\\
       b(t)\\
     \end{array}
   \right]=e^{it \sqrt{\kappa^2-\rho^2}} e^{-\eta/2} M(\eta) 
   \left[
     \begin{array}{cc}
       c_1\\
       c_2 \\
     \end{array}
   \right],
\end{equation*}
where $M(\eta)$ is given by Eq.(\ref{eq:loop-A:matrixM}). By setting $c_2=0$ and using the analytic property of $F^{(n)}$, we can obtain the Floquet state corresponding to $\lambda_2=-\sqrt{\kappa^2-\rho^2}$:
\begin{equation}
  \left[
     \begin{array}{c}
       a_2(t)\\
       b_2(t)\\
     \end{array}
   \right]\sim e^{i\sqrt{\kappa^2-\rho^2}t} e^{-\eta/2} 
   \left[
     \begin{array}{c}
       F^{(0)}\\
       -\frac{\omega}{\kappa}p_1\left( F^{(0)} +\frac{\eta}{p_2} F^{(1)} \right) \\
     \end{array}
   \right]. 
\end{equation}
The Floquet state corresponding to $\lambda_1=\sqrt{\kappa^2-\rho^2}$ can be obtained by requiring 
$[a(t_0+T),b(t_0+T)]^T=e^{-i\sqrt{\kappa^2-\rho^2}T} [a(t_0),b(t_0)]^T$, resulting in 
\begin{equation}
  \left[
     \begin{array}{c}
       a_1(t)\\
       b_1(t)\\
     \end{array}
   \right] \sim e^{i\sqrt{\kappa^2-\rho^2}t} e^{-\eta/2} M(\eta)
   \left[
     \begin{array}{c}
      \frac{2\pi i e^{-iT\sqrt{\kappa^2-\rho^2}}}{\Gamma(p_2)\Gamma(1+p_1-p_2)} \\
       1-e^{-2iT\sqrt{\kappa^2-\rho^2}} \\
     \end{array}
   \right].
\end{equation}

\subsection{Results of non-chirality degree}
The non-chirality degree $\chi_c$ in the perfect dynamics can be calculated from the exact solution, while that in noisy dynamics can only computed numerically. According to \textbf{Lessons from Numerical Simulation}, we should use high enough working precision in numerical simulations to avoid additional noise effects brought by roundoff errors.

\begin{figure}
  \centering
  \includegraphics[width=0.9\textwidth]{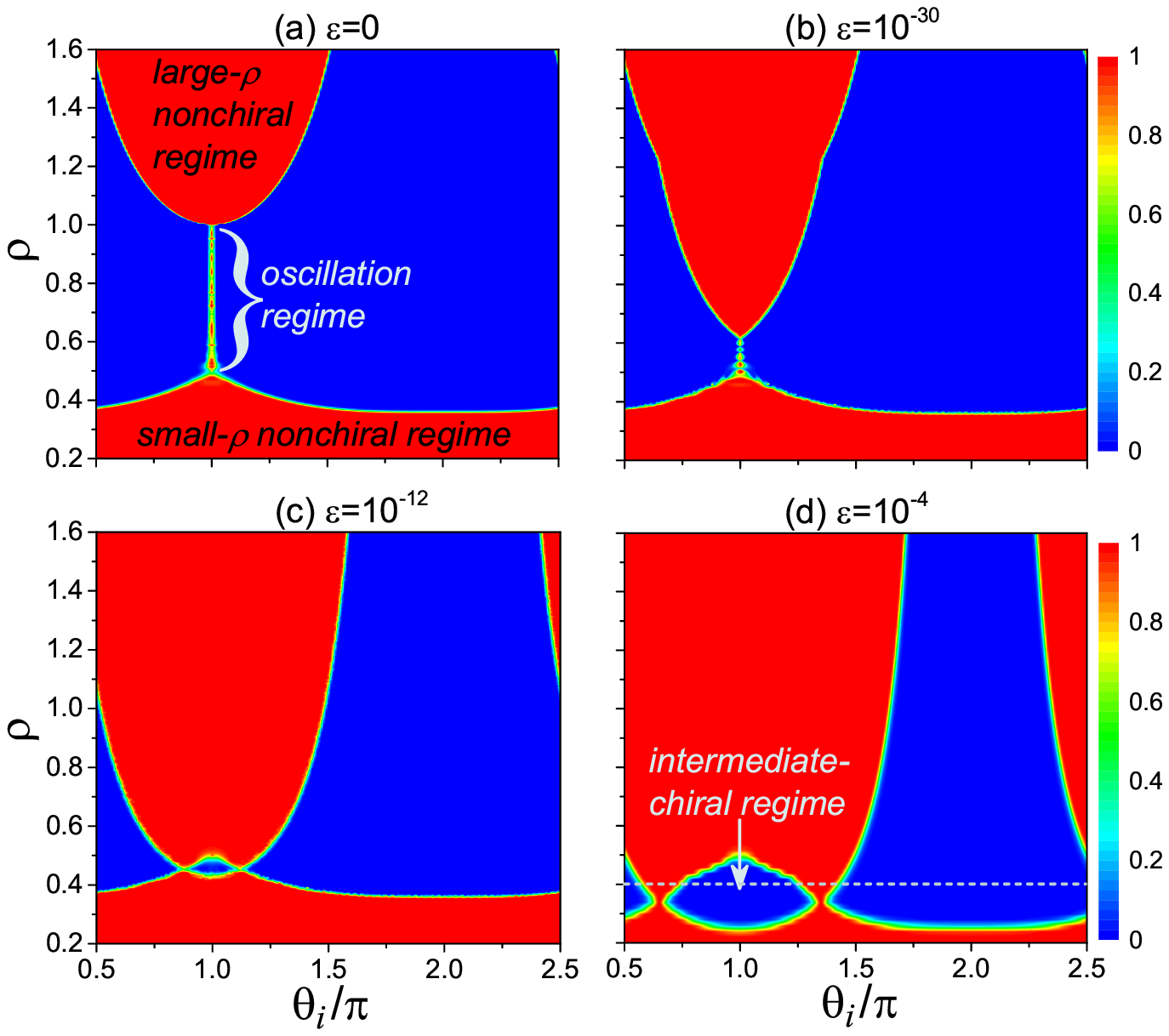}
  \caption{The non-chirality degree $\chi_c$ for Loop-A in the $\rho$-$\theta_i$ plane with fixed $\omega=1/30$ ($\kappa=1$) and (a) $\varepsilon=0$, (b) $\varepsilon=10^{-30}$, (c) $\varepsilon=10^{-12}$ and (d) $\varepsilon=10^{-4}$. The dashed line in (d) indicates the $\rho=0.4$ line, which would be studies in Fig.\ref{fig:s3}.   }\label{fig:s2}
\end{figure}
In Fig.\ref{fig:s2} we plot $\chi_c$ in the $\rho$-$\theta_i$ plane with $\omega=1/30$ ($\kappa=1$) and $\varepsilon=0, 10^{-30}, 10^{-12}$ and $10^{-4}$. In perfect dynamics [Fig.\ref{fig:s2}(a)] we observe three regimes near $\theta_i=\pi$, namely the ``large-$\rho$ nonchiral regime'' for $\rho>1$, the oscillating regime for $0.5<\rho<1$ and the ``small-$\rho$ nonchiral regime'' for $\rho<0.5$. The behavior near $\theta_i=0$ (or equivalently $2\pi$) is quite different: there are only two distinct regimes, namely the chiral regime for large $\rho$ and the nonchiral regime for small $\rho$. However, the critical boundary $\rho_c$ is not easy to derived analytically. As the noise strength increases, the lower boundary of the ``large-$\rho$ nonchiral regime'' move downward and the oscillating regime is reduced [Fig.\ref{fig:s2}(b)]  until vanishing at some noise strength. Thereafter the ``large-$\rho$ nonchiral regime'' overlaps with the ``small-$\rho$ nonchiral regime'' and a strange chiral regime (which we call ``intermediate chiral regime'') emerges in the overlapping region [Fig.\ref{fig:s2}(c) and (d)]. Note that when $\rho<0.5$ the loops do not encircle any EP, where the non-chirality degree still shows complicated behavior in the parameter space.

\begin{figure}
  \centering
  \includegraphics[width=0.6\textwidth]{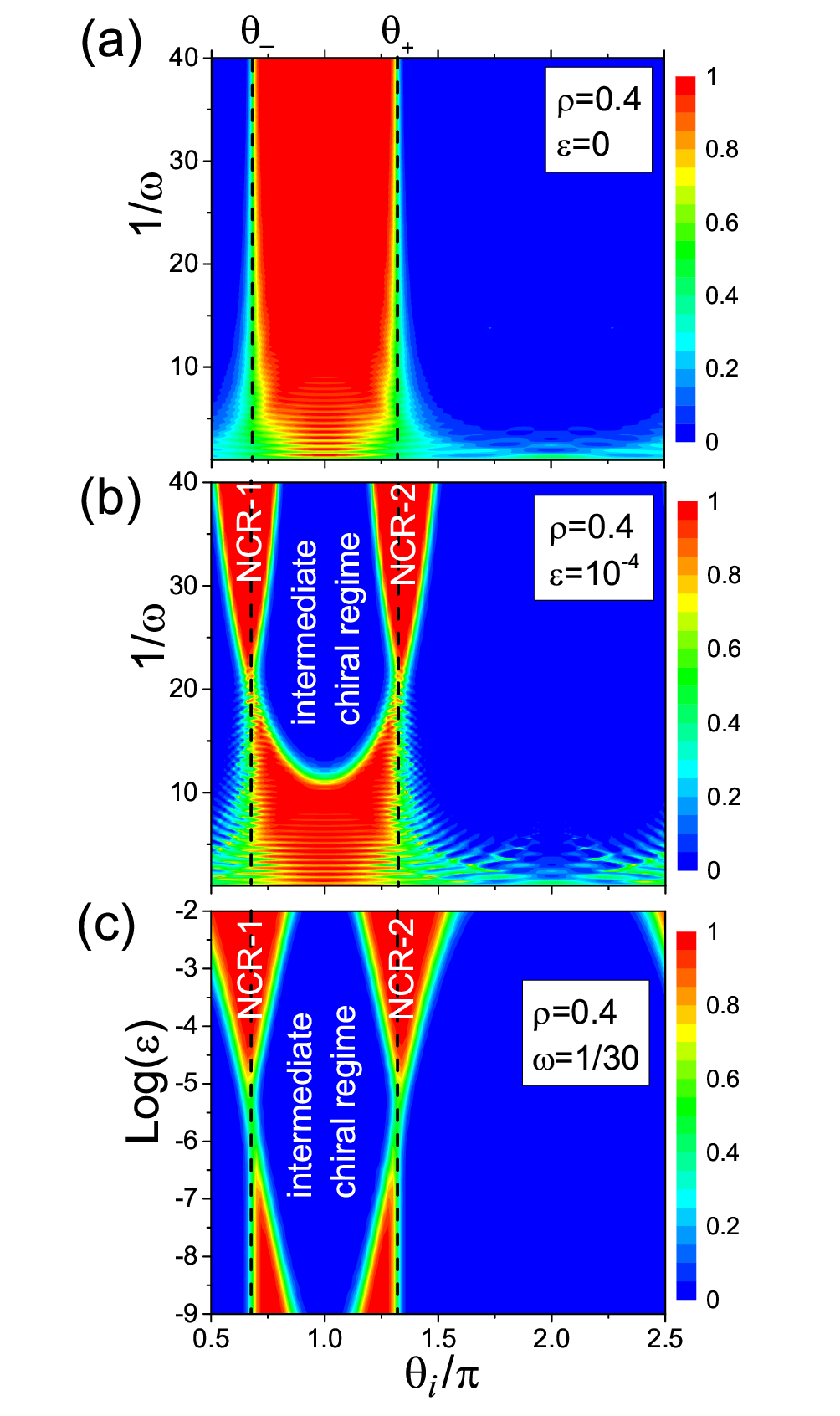}
  \caption{The non-chirality degree $\chi_c$ for Loop-A with $\rho=0.4$ ($\kappa=1$). Both perfect dynamics (a) and noisy dynamics (b) are computed in the $(1/\omega)$-$\theta_i$ plane. The bottom panel (c) shows the $\chi_c$ as a function of the starting point $\theta_i$ and noise strength $\varepsilon$ for fixed speed $\omega=1/30$. The dashed lines depict two critical starting points $\theta_{\pm}$. NCR-1 and NCR-2 denotes two ``nonchiral regimes''.   }\label{fig:s3}
\end{figure}
In Fig.\ref{fig:s3} we plot the non-chirality degree $\chi_c$ as a function of $\omega$ and the starting position $\theta_i$ in both perfect ($\varepsilon=0$) and noisy ($\varepsilon=10^{-4}$) dynamics. In the fast evolution regime $\chi_c$ shows complicated oscillations, while in the slow evolution regime $\chi_c$ is close to either $1$ or $0$, with a sharp transition at some critical point. In the perfect dynamics  [Fig.\ref{fig:s3}(a)] the transition happens at two critical points $\theta_{\pm}$, indicated by dashed lines in the figure. In the noisy dynamics [Fig.\ref{fig:s3}(b)] the critical lines breaks into four, resulting in two nonchiral regimes (named NCR-1 and NCR-2) between which exists the ``intermediate chiral regime'' observed in Fig.\ref{fig:s2}(d). 
NCR-1 and NCR-2 terminate at some critical speed and $\theta_{\pm}$  [Fig.\ref{fig:s3}(b)]. In Fig.\ref{fig:s3}(c) we plot the effect of different noise strength, and also observe NCR-1, NCR-2 and the intermediate chiral regime. By comparing Fig.\ref{fig:s3}(b) and (c) we observe similar behavior in the large $1/\omega$ and the large $\log(\varepsilon)$ regions, which is just a manifestation of \emph{\textbf{the $\omega$-$\varepsilon$ competition}} discovered in the main text.

\begin{figure}
  \centering
  \includegraphics[width=0.9\textwidth]{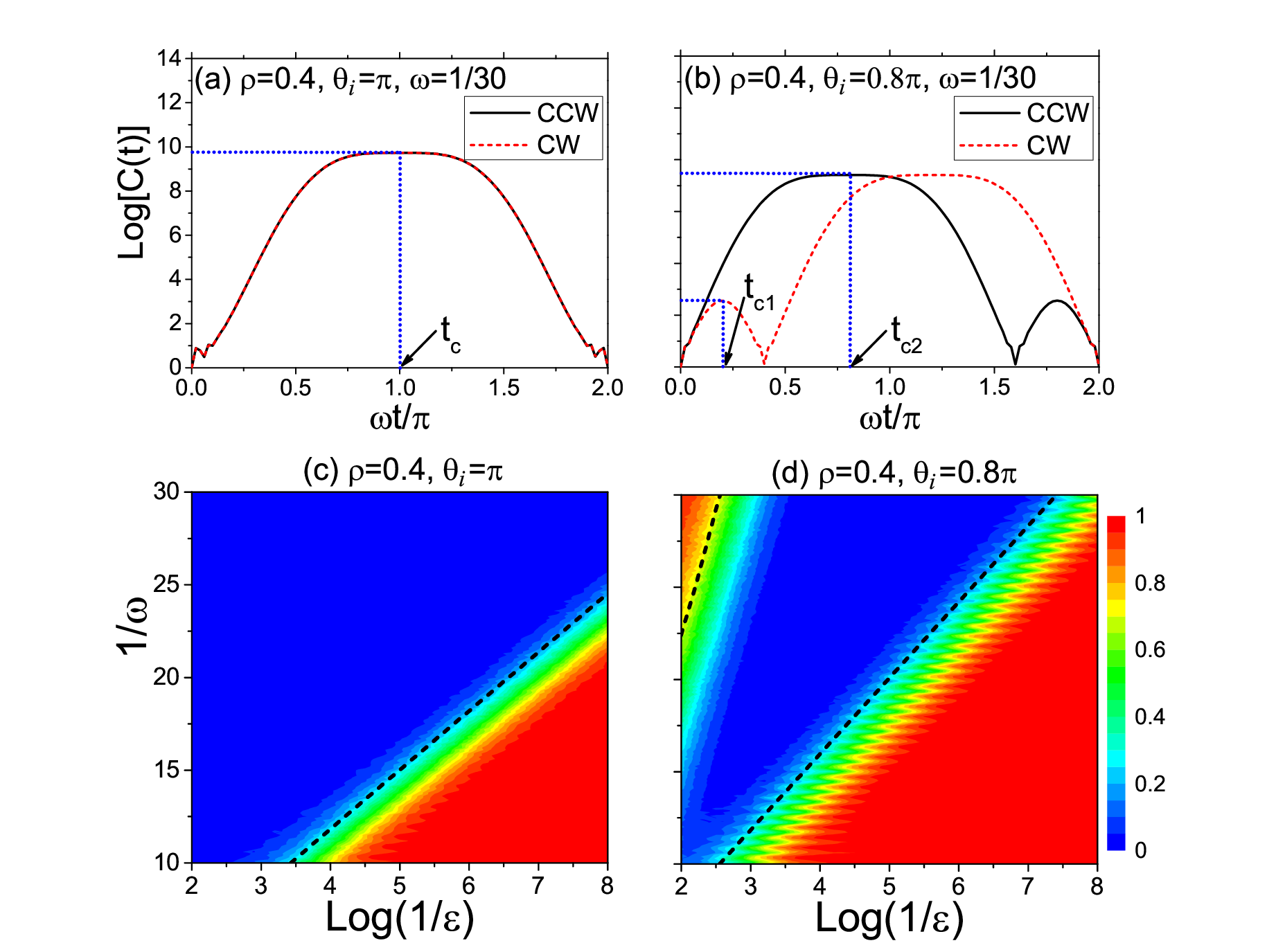}
  \caption{The non-chirality degree $\chi_c$ in the $\log(1/\varepsilon)$-$(1/\omega)$ plane and the logarithm of the condition number $C(t)$ for Loop-A with $\rho=0.4$, demonstrating the speed-noise competition similar to that in Fig.4 in the main text.    } \label{fig:s4}
\end{figure}
The speed-noise ($\omega$-$\varepsilon$) competition is studied for $\rho=0.4$ and $\theta_i=\pi, 0.8\pi$ in Fig.\ref{fig:s4}. 
We also observe similar competition behavior as that in the main text, although now the loops for $\rho=0.4$ do not encircle any EP. The critical boundary still has the scaling law $1/\omega\sim \log(1/\varepsilon)$ [Fig.\ref{fig:s4}(c) and (d)]. However, for $\theta_i=0.8\pi$ there exist two critical boundaries both of which satisfies the above scaling law. Such boundaries could still explained from the condition number $C(t)$ of the transfer matrix [Fig.\ref{fig:s4}(a) and (b)].

\section{Another Series of Loops: Loop-B}
\subsection{The Transfer Matrix, the Floquet States}
In Loop-B we have $h_z=-i\rho e^{i\omega t}$. In this section we use $t$ instead of $\theta=\omega t$ as the arguments of the transfer matrix $S$. Then the general solution of the state vector $(a,b)^T$ can be expressed as
\begin{equation}\label{eq:Loop-B:abMc}
  \left[\begin{array}{cc}
    a(t)  \\
    b(t)
  \end{array} \right] =e^{i\kappa t} e^{-\eta/2}  M(t) \left[ \begin{array}{cc}
    c_1  \\
    c_2
  \end{array}\right],
\end{equation}
where
\begin{equation}
  M(t)=\left[
         \begin{array}{cc}
           F^{(0)} & \quad U^{(0)} \\
           -\left(F^{(0)} +\frac{\eta}{1+2\kappa/\omega} F^{(1)} \right) & \quad -U^{(0)} +\eta U^{(1)} \\
         \end{array}
       \right].
\end{equation}
The coefficients $c_1,c_2$ are determined by the initial condition at $t=t_0$:
\begin{equation}
  \left[\begin{array}{cc}
    a_0  \\
    b_0
  \end{array} \right] =e^{i\kappa t_0} e^{-\eta_0/2}  M_0 \left[ \begin{array}{cc}
    c_1  \\
    c_2
  \end{array}\right],
\end{equation}
where $(a_0,b_0)\equiv (a(0), b(0))$ and $M_0\equiv M(t=0)$. Then
\begin{equation}
  \left[ \begin{array}{cc}
    c_1  \\
    c_2
  \end{array}\right] =e^{-i\kappa t_0} e^{\eta_0/2}  M_0^{-1} \left[\begin{array}{cc}
    a_0  \\
    b_0
  \end{array} \right],
\end{equation}
and hence the transfer matrix $S(t)$ reads
\begin{equation}
  S(t,t_0)=e^{i\kappa (t-t_0)} e^{-(\eta-\eta_0)/2}  M(t) M_0^{-1}.
\end{equation}
The inverse of $M_0$ can be expressed in terms of its determinant $\det{M_0}$ and adjoint matrix $M_0^{\ddag}$, $M_0^{-1}=\frac{1}{\det{M_0}} M_0^{\ddag}$, where
\begin{eqnarray}
  \det{M_0} &=& F_0^{(0)} \left(-U_0^{(0)}+\eta_0 U_0^{(1)} \right) +U_0^{(0)} \left(F_0^{(0)}+ \frac{\eta_0}{1+2\kappa/\omega} F_0^{(1)} \right) \nonumber \\
  &=& \eta_0 \left(F_0^{(0)}  U_0^{(1)} +\frac{1}{1+2\kappa/\omega} F_0^{(1)}  U_0^{(0)} \right) \nonumber \\
  &=& \frac{\Gamma(1+2\kappa/\omega)}{\Gamma(1+\kappa/\omega)} \eta_0^{-2\kappa/\omega} e^{\eta_0}, \\
 M_0^{\ddag} &=& \left[
         \begin{array}{cc}
          -U^{(0)}_0 +\eta_0 U^{(1)}_0 & \quad -U^{(0)}_0 \\
           F^{(0)}_0 +\frac{\eta_0}{1+2\kappa/\omega} F^{(1)}_0  & \quad F^{(0)}_0  \\
         \end{array}
       \right].
\end{eqnarray}
Then the transfer matrix can be rewritten as
\begin{equation}
  S(t,t_0)=e^{i\kappa (t-t_0)} e^{-(\eta+\eta_0)/2} \eta_0^{2\kappa/\omega} \frac{\Gamma(1+\kappa/\omega)}{\Gamma(1+2\kappa/\omega)}  M(t) M_0^{\ddag}.
\end{equation}

Using the analytic continuation as $\theta_0\rightarrow \theta_0+2\pi$:
\begin{align*}
  U^{(0)} &\rightarrow U^{(0)} e^{-4\pi i \kappa/\omega}- \frac{2\pi i e^{-2\pi i\kappa/\omega}} {\Gamma(1+2\kappa/\omega)\Gamma(-\kappa/\omega)} F^{(0)}, \\
  U^{(1)} &\rightarrow U^{(1)} e^{-4\pi i \kappa/\omega} + \frac{2\pi i e^{-2\pi i\kappa/\omega}} {(1+2\kappa/\omega)\Gamma(1+2\kappa/\omega)\Gamma(-\kappa/\omega)} F^{(1)},
\end{align*}
the transfer matrix in one period reads
\begin{align*}
 & S(t_0+T,t_0) = e^{i\kappa T} M(t_0+T) M_0^{-1}\\
  &=e^{i\kappa T} \left\{ M_0 +\left[
         \begin{array}{cc}
         0 & U^{(0)}_0(e^{-4\pi i \frac{\kappa}{\omega}}-1) - \frac{2\pi i e^{-2\pi i\kappa/\omega}} {\Gamma(1+2\kappa/\omega)\Gamma(-\kappa/\omega)} F^{(0)} \\
         0 & (-U^{(0)}_0+\eta_0U^{(1)}_0) (e^{-4\pi i \frac{\kappa}{\omega}}-1) +\frac{2\pi i e^{-2\pi i\kappa/\omega}} {\Gamma(1+2\frac{\kappa}{\omega})\Gamma(-\frac{\kappa}{\omega})} (F^{(0)}_0+\frac{\eta_0}{1+2\kappa/\omega} F^{(1)}) \\
         \end{array}
       \right] \right\} M_0^{-1} \\
  &=e^{i\kappa T} \left\{ \mathds{I}+\left[
         \begin{array}{cc}
         0 & U^{(0)}_0(e^{-4\pi i \frac{\kappa}{\omega}}-1) - \frac{2\pi i e^{-2\pi i\kappa/\omega}} {\Gamma(1+2\kappa/\omega)\Gamma(-\kappa/\omega)} F^{(0)} \\
         0 & (-U^{(0)}_0+\eta_0U^{(1)}_0) (e^{-4\pi i \frac{\kappa}{\omega}}-1) +\frac{2\pi i e^{-2\pi i\kappa/\omega}} {\Gamma(1+2\frac{\kappa}{\omega})\Gamma(-\frac{\kappa}{\omega})} (F^{(0)}_0+\frac{\eta_0}{1+2\kappa/\omega} F^{(1)}) \\
         \end{array}
       \right]  M_0^{-1} \right\}. \\
\end{align*}

From the above exact solution we observe a strange property in this series of loops:  
\begin{itemize}
  \item \textbf{When $\kappa/\omega$ is an integer, $\Gamma(-\kappa/\omega)\rightarrow\infty$ and hence $S(t_0+T,t_0) \rightarrow \mathds{I}$ and $\chi_c=1$  exactly.}
\end{itemize}

The quasienergies are $\lambda_{1,2}=\pm\kappa$. The corresponding Floquet states could be derived from the exact solution in the following way. In Eq.(\ref{eq:Loop-B:abMc}), setting $c_2=0$ gives
\begin{equation*}
  [a,b]^T=e^{i\kappa t} e^{-\eta/2} \left[F^{(0)}, -F^{(0)}-\frac{\eta}{1+2\kappa/\omega} F^{(1)} \right]^T.
\end{equation*}
After one-period evolution, $t_0\rightarrow t_0+T$, we have
\begin{equation*}
  [a(t_0+T),b(t_0+T)]^T= e^{i\kappa T} [a(t_0),b(t_0)]^T.
\end{equation*}
So one quasienergy is $-\kappa$ and the corresponding Floquet mode is just given by the above $[a,b]^T$.

The other quasienergy should be $\kappa$. The corresponding Floquet mode should also takes the form of  Eq.(\ref{eq:Loop-B:abMc})
\begin{equation*}
  \left[
     \begin{array}{cc}
       a(t) \\
       b(t) \\
     \end{array}
   \right]=e^{i\kappa t} e^{-\eta/2} M(t) \left[
     \begin{array}{cc}
      c_1 \\
      c_2 \\
     \end{array}
   \right].
\end{equation*}
By setting $[a(t_0+T),b(t_0+T)]^T= e^{-i\kappa T} [a(t_0),b(t_0)]^T$, we obtain 
\begin{equation*}
  \left[
     \begin{array}{cc}
      c_1 \\
      c_2 \\
     \end{array}
   \right]\sim \left[
     \begin{array}{cc}
       \frac{2\pi i e^{-2\pi i\kappa/\omega}}{\Gamma(1+2\kappa/\omega)\Gamma(-\kappa/\omega)} \\
        1-e^{-4\pi i\kappa/\omega}\\
     \end{array}
   \right].
\end{equation*}
 
\subsection{Results of non-chirality degree}
\begin{figure}
  \centering
  \includegraphics[width=0.9\textwidth]{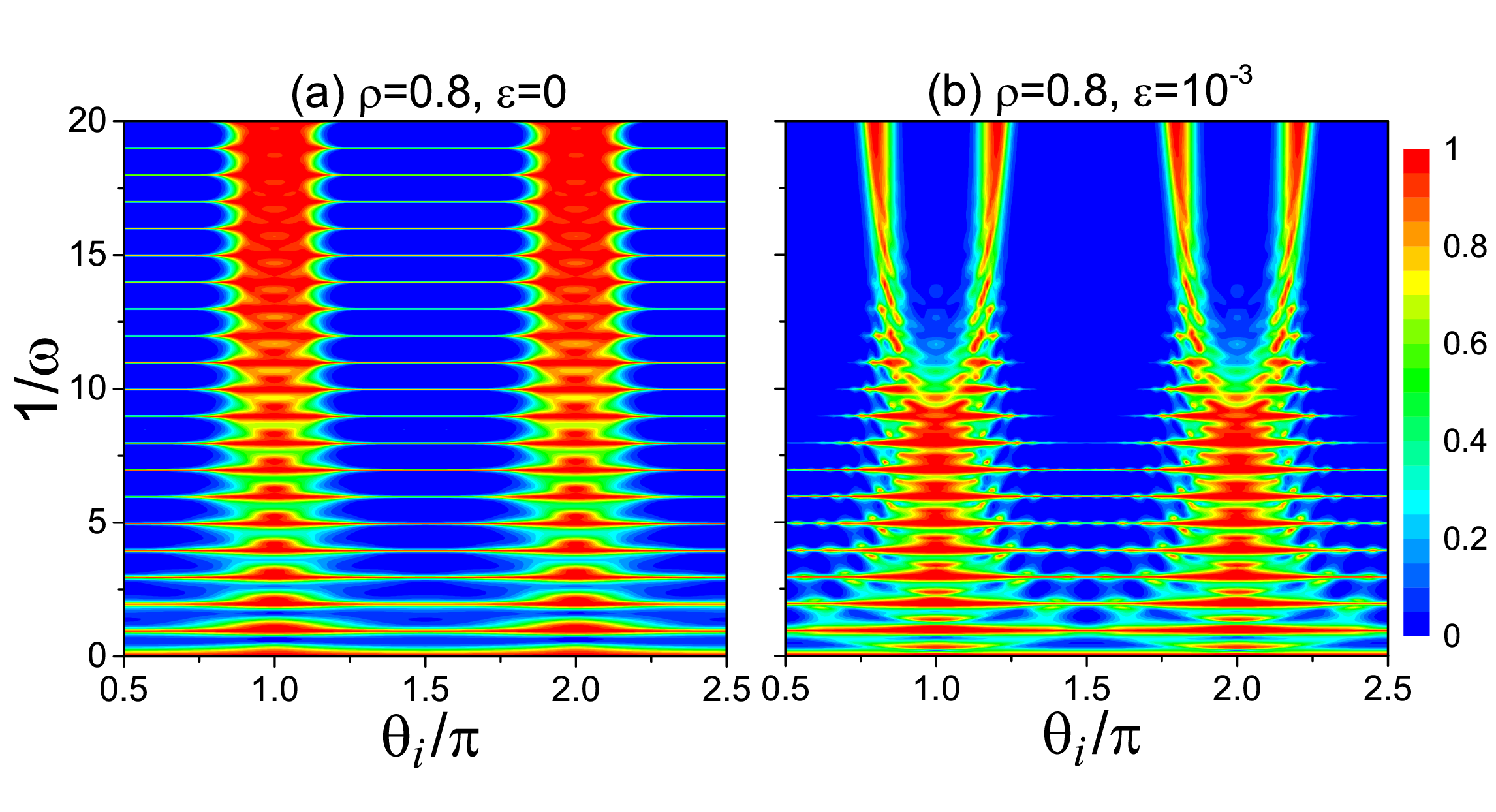}
  \caption{The non-chirality degree $\chi_c$ for Loop-B in the $1/\omega$-$\theta_i$ plane with $\rho=0.8$ ($\kappa=1$) and (a) $\varepsilon=0$, (b) $\varepsilon=10^{-3}$.   }\label{fig:s5}
\end{figure}
In Fig.\ref{fig:s5} we plot the non-chirality degree $\chi_c$ for Loop-B with $\rho=0.8$. In perfect dynamics ($\varepsilon=0$) and large $\omega$, the non-chirality degree $\chi_c$ shows oscillation behavior for arbitrary starting points $\theta_i$. However, for small $\omega$, the oscillations near $\theta_i=0,\pi$ vanish and are replaced by nonchiral regimes ($\chi_c\approx1$) [Fig.\ref{fig:s5}(a)]. A notable feature in the perfect dynamics is that when $\kappa/\omega$ is an integer, the non-chirality degree $\chi_c=1$ exactly. However, as the noise strength increases, the dynamical behavior with slow speed changes significantly. When $\varepsilon=10^{-3}$, the chirality oscillation still exist in the fast (small $1/\omega$)  region, while in the slow (large $1/\omega$) region, only four nonchiral regimes and four chiral regimes survive [Fig.\ref{fig:s5}(b)], with sharp transitions between them.

\begin{figure}
  \centering
  \includegraphics[width=0.9\textwidth]{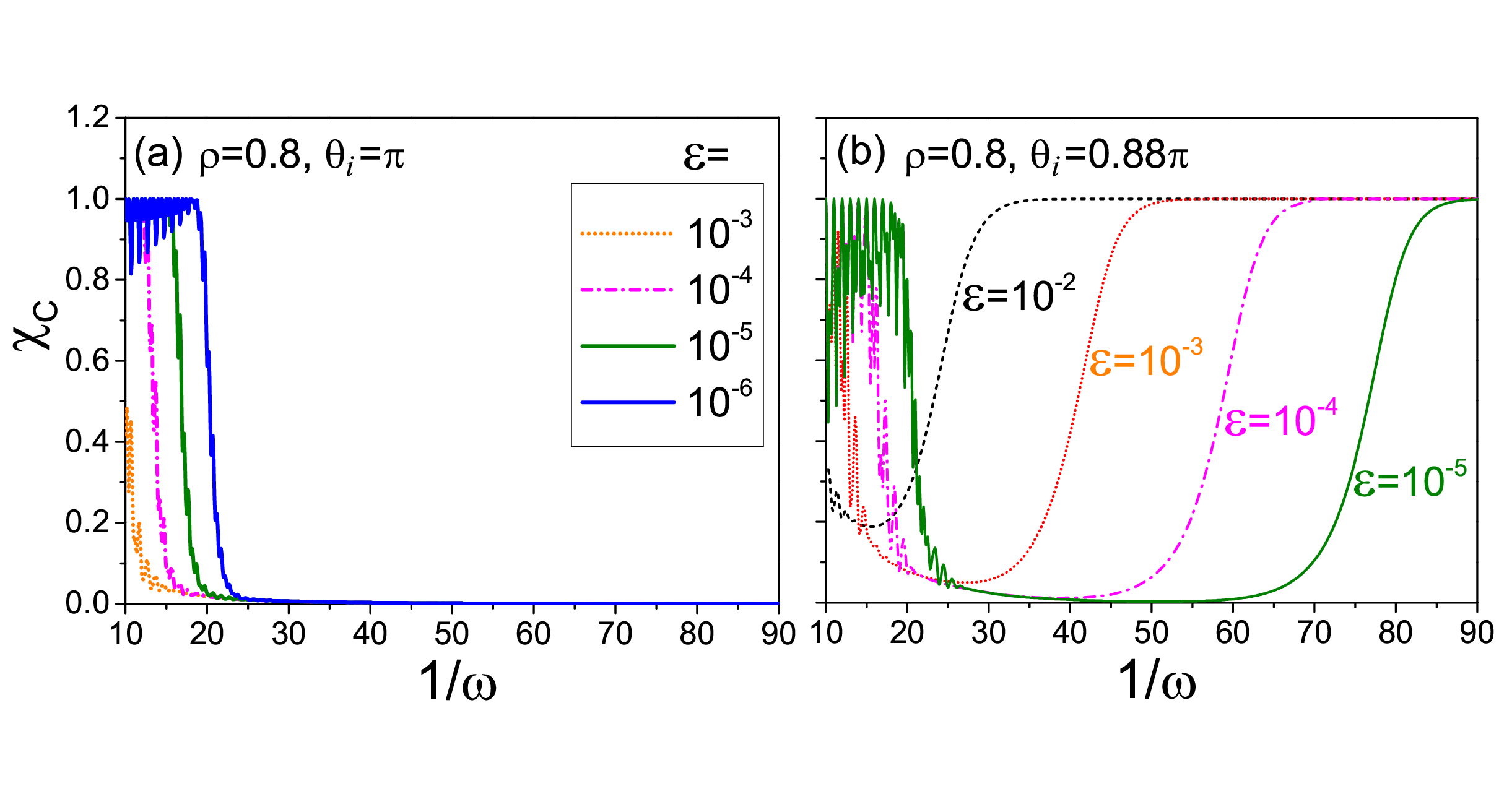}
  \caption{The non-chirality degree as a function of $1/\omega$ for Loop-B with $\rho=0.8$ and (a) $\theta_i=\pi$, (b) $\theta_i=0.88\pi$.     } \label{fig:s6}
\end{figure}
To further study the dynamical behavior in the slow evolution limit, in Fig.\ref{fig:s6} we plot the non-chirality degree as a function of $1/\omega$ in a larger range with $\rho=0.8$ and $\theta_i=\pi,0.88\pi$. When $\theta_i=\pi$ [Fig.\ref{fig:s6}(a)], we could  identify two distinct regimes: the fast regime where $\chi_c$ oscillates and the slow regime where $\chi_c\approx 0$. The critical boundary moves toward larger $1/\omega$ as $\log(1/\varepsilon)$ increases, indicating a scaling law $\log(1/\varepsilon)\sim (1/\omega)$.  When $\theta_i=0.88\pi$ [Fig.\ref{fig:s6}(b)], we could  identify three distinct regimes as $\varepsilon$ is very small: the fast regime where $\chi_c$ oscillates, the intermediate regime where $\chi_c\approx0$ and the slow regime where $\chi_c\approx 1$. The two critical boundaries move toward larger $1/\omega$ as $\log(1/\varepsilon)$ increases, indicating that both of them satisfy a scaling law $\log(1/\varepsilon)\sim (1/\omega)$, similar to Fig.\ref{fig:s4}(d).

\begin{figure}
  \centering
  \includegraphics[width=0.9\textwidth]{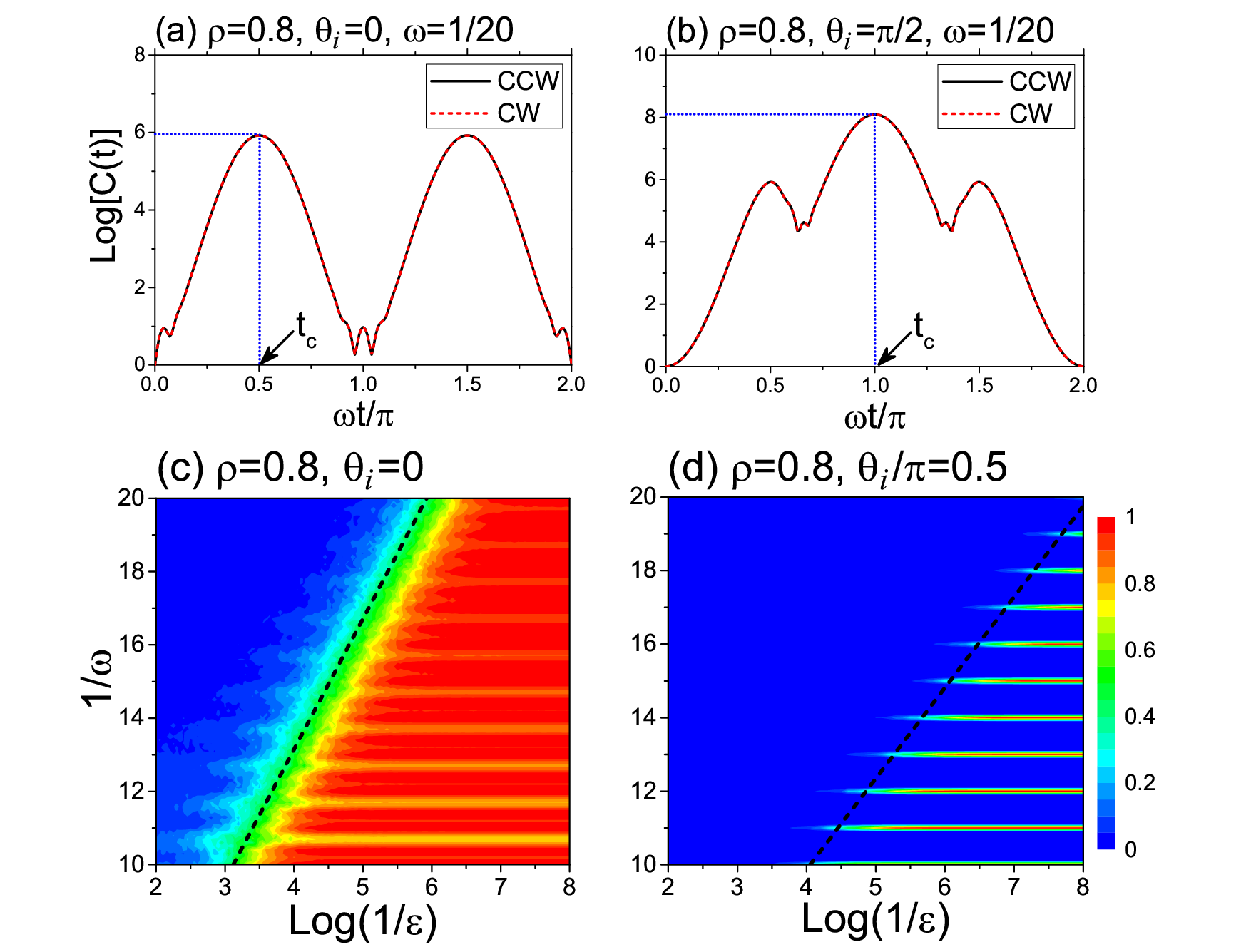}
  \caption{The non-chirality degree $\chi_c$ in the $\log(1/\varepsilon)$-$(1/\omega)$ plane and the logarithm of the condition number $C(t)$ for Loop-B with $\rho=0.8$, demonstrating the speed-noise competition similar to that in Fig.4 in the main text.     } \label{fig:s7}
\end{figure}
In Fig.\ref{fig:s7} we study the speed-noise competition and the critical boundary in Loop-B for $\rho=0.8$ and $\theta_i=0,\pi/2$. We could also find the scaling law $\log(1/\varepsilon)\sim(1/\omega)$ its relation to the condition number $C(t)$ of the transfer matrix, supporting the conclusion in the main text.

\section{Conjecture on How to Find the $\omega$-$\varepsilon$ Competition Critical Boundary from the Condition Number $C(t)$ }
We see that not all maximums of $C(t)$ correspond to critical boundary in the  $\omega$-$\varepsilon$ competition plane. But all critical boundaries  correspond to some maximum of $C(t)$. From Fig.4 in the main text and Fig.\ref{fig:s4} and Fig.\ref{fig:s7}, we make the following \textbf{\emph{conjecture}} about how to find the $\omega$-$\varepsilon$ competition critical boundary from the condition number $C(t), t\in[0, T=2\pi/\omega]$. 

(i) Find the max local maximum of $C(t)$ in $(0,T)$ (note that not $[0,T]$) at $t_{c,0}<T$. If there exists one time $t_\ast\in(t_{c,0},T)$ such that $C(t_\ast)\approx 1$, then $\varepsilon_c C(t_{c,0})=1$ corresponds to one critical boundary in the   $\log(1/\varepsilon)$-$(1/\omega)$ plane. 

(ii) Then find the max local maximum of $C(t)$ in $(0,t_{c,0})$ at $t_{c,1}<t_{c,0}$. If there exists one time $t_\ast\in(t_{c,1},t_{c,0})$ such that $C(t_\ast)\approx 1$, then $\varepsilon_c C(t_{c,1})=1$ corresponds to one critical boundary in the $\log(1/\varepsilon)$-$(1/\omega)$ plane. 

(iii) Repeat (ii) until all critical times are found.

One can check this conjecture for the cases given in Fig.4 in the main text and Fig.\ref{fig:s4}, Fig.\ref{fig:s7}. We also add more results for the loops in the main text in Fig.\ref{fig:s8}, supporting the above conjecture. In  Fig.\ref{fig:s8}(a) we note that the chirality is especially stable with respect to noise when the loop starts from the symmetric phase. This has been experimentally verified recently by Lu \emph{et al.} [Commun. Phys. 8, 91 (2025)]. However, they did not investigate the stability for loops starting from general points. 

\begin{figure}
  \centering
  \includegraphics[width=\textwidth]{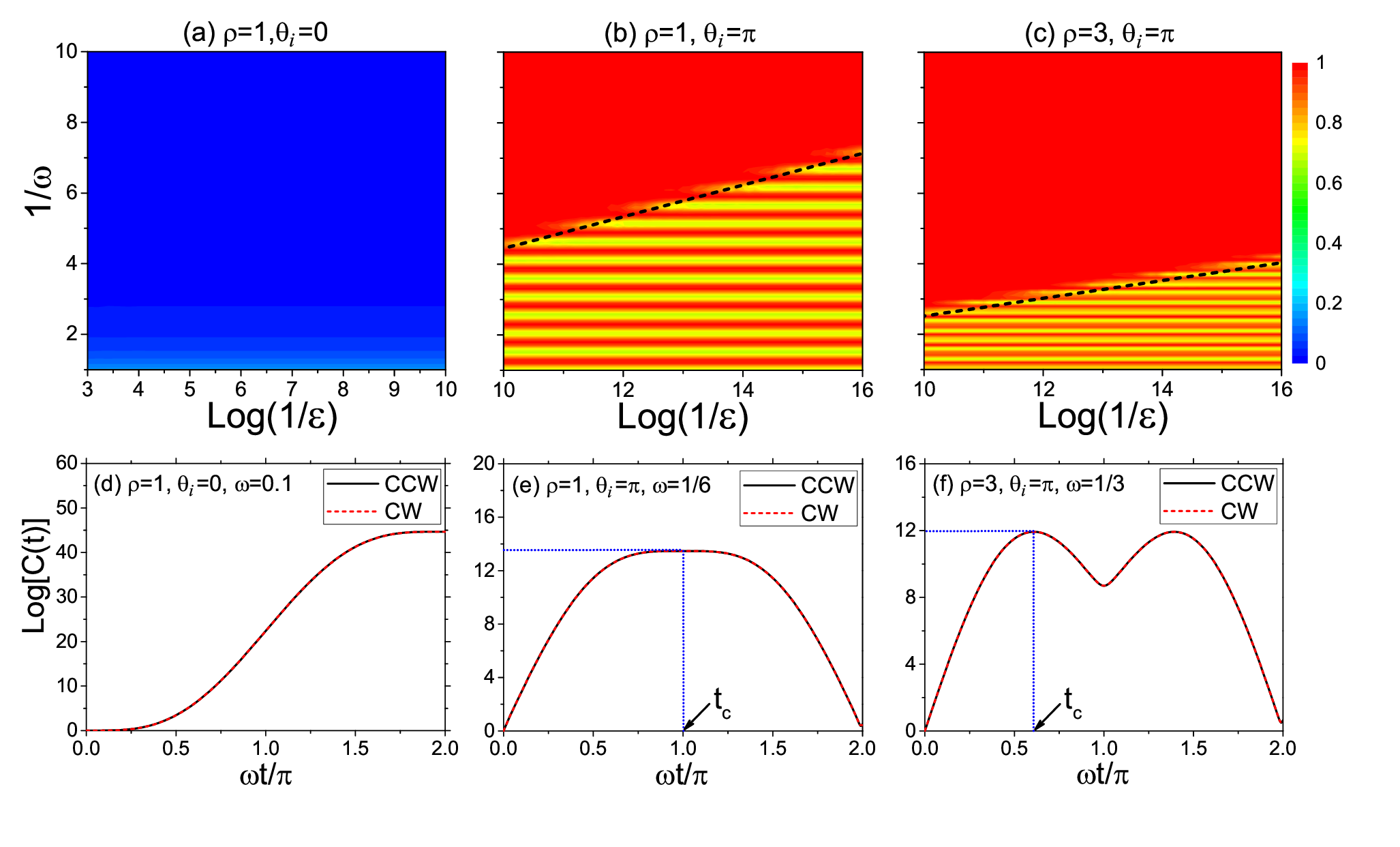}
  \caption{More plots about the speed-noise competition for loops in the main text (supplementary to Fig.4 in the main text). } \label{fig:s8}
\end{figure}

\section{Remarks on PRX8,021066 (2018)}
We make the following remarks on PRX8,021066 (2018) 

(1) The solution is correct, but the asymptotic analysis is wrong. Correct asymptotic analysis is given in our main text [Eqs.(10a,10b)]. 

(2) The nonchiral dynamics ($\chi_c=1$) starting from points in the broken phase is a result of noise but not an intrinsic property of the encircling loops. There is a crossover from oscillating $\chi_c$ (perfect limit) to $\chi_c=1$ (noisy limit) in the $\omega$-$\varepsilon$ plane, where $\varepsilon$ is a measure of the noise strength. 
This underlying physical reason of the nonchiral dynamics was not recognized in PRX8,021066 (2018). Instead, they mistook the physics as a result of the starting point (in the broken phase). In fact, the starting point just affect the sensitivity of the dynamics to noise and hence the critical boundary in the $\omega$-$\varepsilon$ plane.

(3) The experimental results correspond to the noisy limit due to unavoidable noise in their experimental setup. 
They should be able to observe the transition from the nonchiral dynamics in the noisy limit to the chirality oscillation in the perfect limit by increasing the encircling speed with fixed loop radius and underlying noise strength.


\end{document}